\documentclass[letterpaper,10pt]{article}
\usepackage{graphicx}	
\usepackage{amssymb}
\usepackage{pifont}
\usepackage{bm}
\usepackage{amsmath}
\usepackage{amsthm}
\usepackage{algorithm,algorithmic}
\usepackage{enumitem}   
\usepackage{tabulary}
\usepackage{booktabs}
\usepackage{float}
\usepackage[english]{babel}
\usepackage{graphicx}
\usepackage{subcaption}
\usepackage{tikz}
\usepackage[]{cryptocode}
\usepackage{array}
\usepackage{fancyhdr}
\newcolumntype{L}[1]{>{\raggedright\let\newline\\\arraybackslash\hspace{0pt}}m{#1}}
\newcolumntype{C}[1]{>{\centering\let\newline\\\arraybackslash\hspace{0pt}}m{#1}}

\let\oldbibliography\thebibliography
\renewcommand{\thebibliography}[1]{%
    \oldbibliography{#1}%
    \setlength{\itemsep}{-1pt}%
}

\newtheorem{prop}{Proposition}[section]
\newtheorem{lem}{Lemma}[section]

\newcommand{\descr}[1]{\smallskip \noindent \textbf{#1}}
\newcommand{\descrintro}[1]{\vspace{0.8pt} \noindent \textbf{#1}}

\newcommand{\etal}{\textit{et al.\,}}
\newcommand{\sysname}{CRISP\;}
\newcommand{\syscomma}{CRISP}

\newcommand{\file}{\includegraphics[scale=0.02]{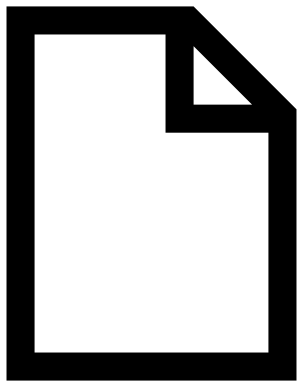}}%
\newcommand{\lock}{\includegraphics[scale=0.02]{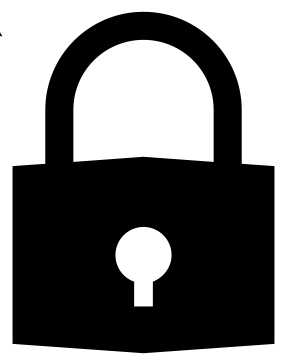}}
\newcommand{\cert}{\includegraphics[scale=0.025]{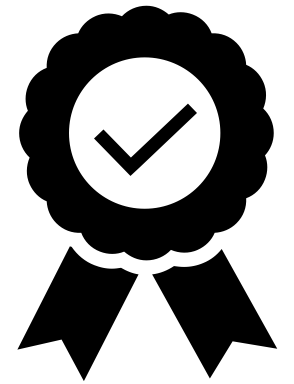}}%
\newcommand{\email}{\sffamily \small \vspace{-8pt}}						

\begin{document}
\title{Privacy and Integrity Preserving Computations with CRISP}

\author{
{\rm Sylvain Chatel}
\and
{\rm Apostolos Pyrgelis}
\and
{\rm Juan Ram\'{o}n Troncoso-Pastoriza}
\and
{\rm Jean-Pierre Hubaux}
} 

\date{}
\maketitle
\vspace{-0.8cm}
\begin{center}
Laboratory for Data Security, EPFL
\end{center}
\begin{center}
\email{first.last@epfl.ch}
\end{center}
\vspace{0.3cm}

\thispagestyle{fancy}
\renewcommand{\headrulewidth}{0pt}
\fancyhf{}
\renewcommand{\footrulewidth}{0.4pt}
\lfoot{\footnotesize{This work has been accepted and will be presented at USENIX Security 2021}} 

\begin{abstract} 
In the digital era, users share their personal data with service providers to obtain some utility, e.g., access to high-quality services. Yet, the induced information flows raise privacy and integrity concerns. Consequently, cautious users may want to protect their privacy by minimizing the amount of information they disclose to curious service providers. Service providers are interested in verifying the integrity of the users' data to improve their services and obtain useful knowledge for their business. In this work, we present a generic solution to the trade-off between privacy, integrity, and utility, by achieving authenticity verification of data that has been encrypted for offloading to service providers. Based on lattice-based homomorphic encryption and commitments, as well as zero-knowledge proofs, our construction enables a service provider to process and reuse third-party signed data in a privacy-friendly manner with integrity guarantees. We evaluate our solution on different use cases such as smart-metering, disease susceptibility, and location-based activity tracking, thus showing its versatility. Our solution achieves broad generality, quantum-resistance, and relaxes some assumptions of state-of-the-art solutions without affecting performance.
\end{abstract}

\section{Introduction}\label{intro}
In our inter-connected world, people share personal information collected from various entities, networks, and ubiquitous devices (i.e., data sources) with a variety of service providers, in order to obtain access to services and applications. Such data flows, which typically involve a user, a data source, and a service provider (as depicted in Figure~\ref{fig:sherpa}), are common for a wide range of use cases, e.g., smart metering, personalized health, location-based activity tracking, dynamic road tolling, business auditing, loyalty programs, and pay-as-you-drive insurance. However, due to conflicting interests of the involved parties, such data interactions inherently introduce a trade-off between \textit{privacy}, \textit{integrity}, and \textit{utility}.

Some users seek to protect their \textit{privacy} by minimizing the amount of personal information that they disclose to \textit{curious} third-parties. Service providers are interested in maintaining the value obtained from the users' data. To this end, service providers are concerned about verifying the \textit{integrity} of the data shared by their users, i.e., ensure that the user's data has been certified by a trusted, external, data source. Both parties want to obtain some \textit{utility} from these data flows: Service providers want to use the data for various computations that yield useful knowledge for their business or services, and users share part of their data to obtain services and applications.
As users might not know upfront the number and details of the computations, they wish to offload their data once to the service provider and be contacted only to authorize the revelation of the result.
Thus, in this work we present a solution that enables flexible computations on third-party signed data offloaded to a service provider in a privacy and integrity preserving manner.

To illustrate the inherent trade-off between privacy, integrity, and utility, we detail some of the use cases(more use cases are described in Appendix~\ref{app:uc}):

\descrintro{Smart Metering.} Smart meters (i.e., data sources) measure the consumption of a user's household. The data is shared with a service provider (e.g., a different legal entity) for billing and load-balancing analysis. A user's privacy can be jeopardized as energy consumption patterns can reveal her habits~\cite{cavoukian_smartprivacy_2010,kumar2019smart}. The service provider wants guarantees on the data integrity to provide reliable services~\cite{asghar2017smart}. \textit{Malicious} users might  cheat to reduce their bills or disrupt the service provider's analysis.

\descrintro{Disease Susceptibility.} Medical centers and direct-to-consumer services~\cite{23andme_dna_nodate,dnafit}, provide a user with her DNA sequence to improve her health and to customize her treatments. Genomic data can be used for disease-susceptibility tests offered by service providers, e.g., research institutions that seek to form the appropriate cohorts for their studies. The user wants to protect her data as DNA is considered a very sensitive and immutable piece of information for her and her relatives~\cite{erlich2014routes}. Correspondingly, service providers are keen on collecting users' data and verifying its integrity so that they can use it for disease-risk estimation or other types of analyses, e.g., drug-effect prediction or health certificates. Malicious users might tamper with the genomic data they share to disrupt this process and pass a medical examination.

\descrintro{Location-Based Activity Tracking.} A user's wearable device monitors her location by querying location providers. The user then shares this information with service providers, e.g., online fitness social networks~\cite{international_garmin_nodate} or insurance companies~\cite{sanitas} to obtain activity certificates or discount coupons. As location data can reveal sensitive information, e.g., her home/work places or habits~\cite{nytimes_lastnight_2018, hern_fitness_2018}, the user is concerned about her privacy. Service providers want legitimate data to issue activity certificates, provide discounts for performance achievements, and build realistic user profiles. Malicious users might be tempted to modify their data, aiming to claim fake accomplishments and obtain benefits they are not entitled to.

\begin{figure}[h]
    \centering
    \includegraphics[width=.6\textwidth]{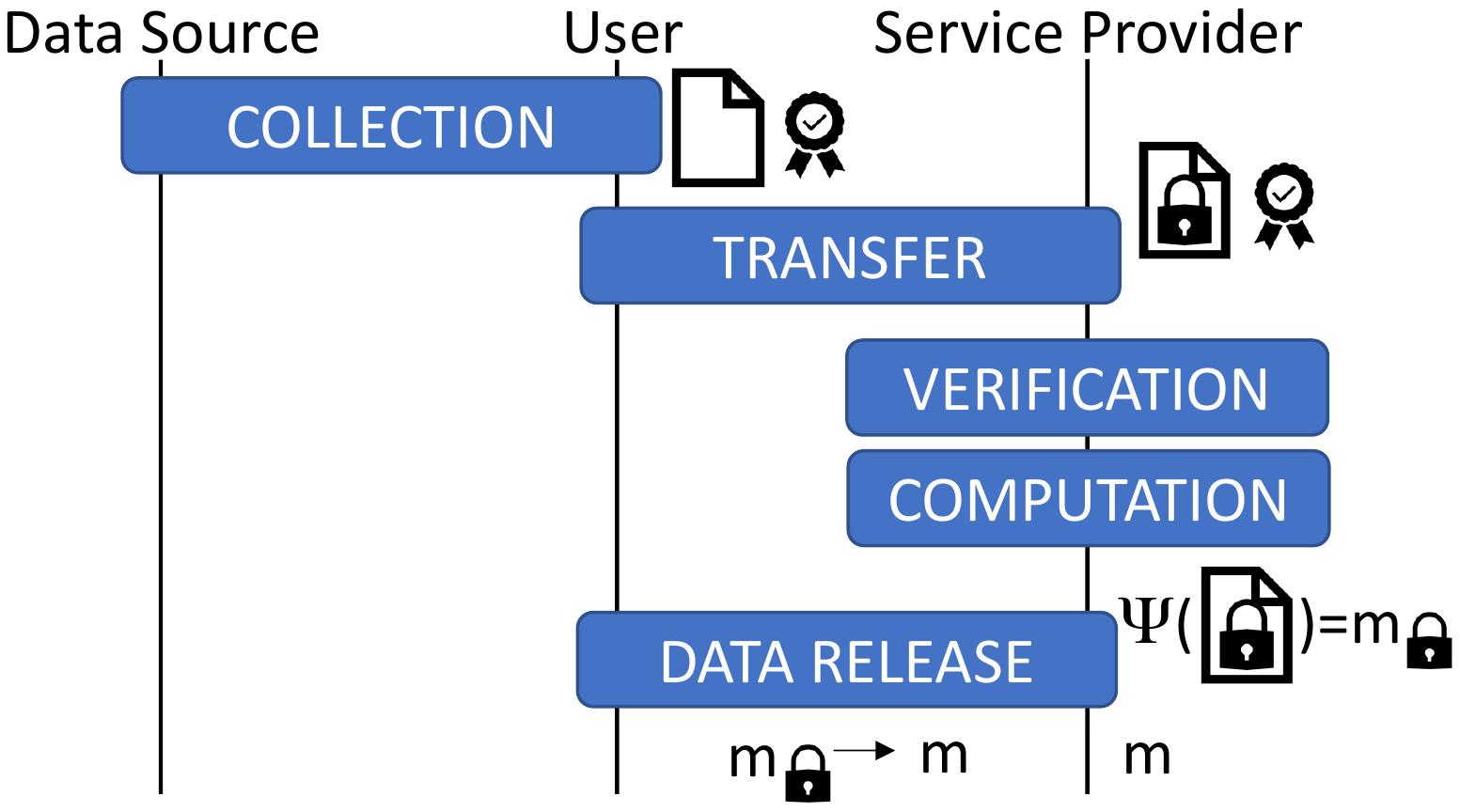}
    \caption{Three-party model and their interaction phases. \protect\file ~is the private information authenticated with \protect\cert. The user protects it via \protect\lock. The service provider computes $\psi(\cdot)$ on the protected data and obtains an output which is revealed as $\text{m}$.}
    \label{fig:sherpa}
\end{figure}

The above use cases fall under the three-party model of Figure~\ref{fig:sherpa}, with (i) malicious users, and (ii) honest-but-curious service providers and data sources; as such, they exhibit the trade-off between privacy, integrity, and utility. To support integrity protection regarding users' data, service providers require a data source to certify it, e.g., by means of a digital signature. This certification should require minimal to no changes to the data source: using only deployed hardware and software infrastructure. Another common denominator is that service providers want to collect users' data and perform various computations. Consequently, users should be able to offload their protected data to service providers (i.e., transfer a copy of the data only once) in such a way that their privacy is preserved, the data integrity can be verified, and various flexible computations are feasible.

A simple solution is to establish a direct communication channel between the data source and the service provider. This way, the data source could compute the operations queried by the service provider on the user's data. However, this would prevent the user from remaining in control of her data and require the data source to bear computations that are outside of its interests. Another approach is to let the data source certify the user's data by using specialized digital signature schemes such as homomorphic signatures~\cite{boneh_homomorphic_2011,catalano_practical_2013,catalano2014generalizing,catalano_security_2018,gorbunov_leveled_2015} or homomorphic authenticators~\cite{ahn_computing_2015,cheon_multi-key_2016,gennaro_fully_2013,matsui_context_2019}. Thus, the user could locally compute the queried operation and provide the service provider with the result and a homomorphic signature attesting its correct computation on her data. However, this would require software modifications at the data source, which would come at a prohibitive cost for existing services, and introduce significant overhead at the user.

In the existing literature, several works specialize in the challenges imposed by the above use cases but provide only partial solutions by either addressing privacy~\cite{ayday_protecting_2013,ChenPDAFT2015,danezis_fast_2014,decristofaro2013secure,lauter_private_2015,li2015pda}, or integrity~\cite{Buchmann2019,chiang2009secure,BamHash,saroiu2009enabling}. The handful of works addressing both challenges require significant modifications to existing hardware or software infrastructures. For instance, SecureRun~\cite{pham_securerun:_2016}, which achieves privacy-preserving and cheat-proof activity summaries, 
requires heavy modifications to the network infrastructure. Similarly, 
smart metering solutions using secure aggregation, 
e.g.,~\cite{abdallah_lightweight_2018,li_preserving_2012,lu2012eppa}, rely on specialized signature schemes that are not yet widely supported by current smart meters. These approaches are tailored to their use case and cannot be easily adapted to others, hence there is the need for a \emph{generic} solution to the trade-off between privacy and integrity, without significantly degrading utility.

ADSNARK~\cite{backes_adsnark:_2015} is a generic construction that could be employed to address the trade-off between privacy, integrity, and utility. In particular, it enables users to locally compute on data certified by data sources and to provide proof of correct computation to service providers. However, ADSNARK does not support the feature of data offloading that enables service providers to reuse the collected data and to perform various computations. Indeed, ADSNARK and other zero-knowledge solutions~\cite{fournet_zql:_nodate,fredrikson2014zo,ben-sasson_scalable_2019}, require the user to compute a new proof every time the service provider needs the result of a new computation. Furthermore, it requires a trusted setup, and is not secure in the presence of quantum adversaries~\cite{katz_improved_2018}. The latter should be taken into account considering recent advances in quantum computing~\cite{arute_quantum_2019} and the long term sensitivity of some data.

In this work we propose \sysname (privaCy and integRIty preServing comPutations),
a novel solution that achieves utility, privacy, and integrity; it is generic, supports data offloading with minimal modification to existing infrastructures, relaxes the need for a trusted setup, and is quantum-resistant. Motivated by the need to protect users' privacy and by the offloading requirement to support multiple computations on their data, \sysname relies on quantum-resistant lattice-based approximate homomorphic encryption (HE) primitives~\cite{cheon_homomorphic_2017} that support flexible polynomial computations on encrypted data without degrading utility. To ensure data integrity, we employ lattice-based commitments~\cite{baum_more_2018} and zero-knowledge proofs~\cite{chase_post-quantum_2017} based on the multi-party-computation-in-the-head (or MPC-in-the-head) paradigm~\cite{ishai_zero-knowledge_2009}, which enable users to simultaneously convince service providers about the correctness of the encrypted data, as well as the authenticity of the underlying plaintext data, using the deployed certification mechanism.

We evaluate our solution on three use cases covering a wide range of applications and computations: smart metering, disease susceptibility, and location-based activity-tracking. Our experimental results show that our construction introduces acceptable computation overhead for users to privately offload their data and for service providers to both verify its authenticity and to perform the desired computations. The magnitude of the communication overhead fluctuates between tens and hundreds of mega bytes per proof and is highly dependent on the use case and its security requirements. To this end, in Section~\ref{eval}, we also present different optimizations that can reduce the proof size, thus making our construction practical for real-life scenarios. Additionally, we demonstrate that \sysname achieves high accuracy in the computations required by the use cases, yielding an average absolute accuracy of more than $99.99$\% over the respective datasets. Compared to the state of the art~\cite{backes_adsnark:_2015}, we reach comparable performance and achieve post-quantum security guarantees with more flexibility in the computations.

Our contributions are the following:
\begin{itemize}
\itemsep-0.5em 
\item A generic, quantum-resistant solution that enables privacy and integrity preserving computations in the three-party model of Figure~\ref{fig:sherpa}, with minimal modifications of the existing infrastructure;
\item the necessary primitives to achieve authenticity verification of homomorphically encrypted data in the quantum random oracle model;
\item an implementation of \sysname~\cite{codeRepo} and its performance evaluation on various representative use cases that rely on different types of computations and real-world datasets.
\end{itemize}
To the best of our knowledge, it is the first time such a solution is proposed.

This paper is organized as follows: In Section~\ref{model}, we discuss the system and threat model on which our construction operates. In Section~\ref{prelims}, we introduce useful cryptographic primitives. Then, we present \syscomma's architecture in Section~\ref{archi} and in Section~\ref{analysis} we perform its privacy and security analysis. In Section~\ref{eval}, we evaluate \sysname on various use cases and in Section~\ref{discussion} we discuss some of its aspects. We review the related work in Section~\ref{related} and conclude in Section~\ref{conc}.

\section{Model}\label{model}

We describe the model, assumptions, and objectives of \syscomma.%

\subsection{System Model}\label{model:system}

We consider three entities: a user, a service provider, and a data source, as depicted in Figure~\ref{fig:sherpa}. The user obtains from the data source certified data about herself and/or her activities, she subsequently shares it with the service provider to obtain some service. The user is interested in sharing (i.e., offloading) her data while protecting her privacy, i.e., she wants to have full control over it but still obtain utility from the service provider. The service provider is interested in (i) verifying the authenticity of the user's data, and (ii) performing on it multiple computations that are required to provide the service and/or improve its quality. The data source can tolerate only minimal changes to its operational process and cannot cope with any heavy modification to the underlying infrastructure and dependencies of the hardware and software. Finally, we assume the existence of a public key infrastructure that verifies the identities of the involved parties as well as secure communication channels between the user and the data source, and between the user and the service provider.

%

\subsection{Threat Model}\label{model:threat}
We present the assumed adversarial behavior for the three entities of our model with computationally bounded adversaries.

\descr{Data Source.} The data source is considered honest and is trusted to generate valid authenticated data about the users' attributes or activities.

\descr{Service Provider.} The service provider is considered \textit{honest-but-curious}, i.e., it abides by the protocol and does not engage in denial-of-service attacks. However, it might try to infer as much information as possible from the user's data and perform computations on it without the user's consent. 

\descr{User.} We consider a \textit{malicious but rational} user. In other words, she engages in the protocol and will try to cheat only if she believes that she will not get caught -- and hence be identified and banned -- by the service provider. This type of adversary is also referred to as \emph{covert} in the literature~\cite{aumann2007security}. The user is malicious in that she might try to modify her data, on input or output of the data exchange, in order to influence the outcome of the service provider's computations to her advantage. Nonetheless, the user is rational, as she desires to obtain utility from the service provider and thus engages in the protocol.

\subsection{Objectives}\label{model:objectives}
Overall, the main objective of our construction is to provide the necessary building blocks for secure and flexible computations in the considered three-party model. To this end, user's privacy should be protected by keeping her in control of the data even in a post-quantum adversarial setting, and the service provider's utility should be retained by ensuring the integrity of the processed data. The above objectives should be achieved by limiting the impact on already deployed infrastructures, thus, by requiring only minimal changes to the data source's operational process. More formally, the desired properties are: (a) \textbf{Utility}: Both user and service provider are able to obtain the correct result of a public computation on the user's private data; (b) \textbf{Privacy}: The service provider does not learn anything more than the output of the computation on the user's private data; and (c) \textbf{Integrity}: The service provider is ensured that the computation is executed on non-corrupted data certified by the data source.
%

\section{Preliminaries}\label{prelims}
We introduce the cryptographic primitives used in Section~\ref{archi} to instantiate \syscomma. In the remainder of this paper, let $a \leftarrow \chi$ denote that $a$ is sampled from a distribution $\chi$; a vector be denoted by a boldface letter, e.g., $\textbf{x}$, with $\textbf{x}[i]$ its $i$-th element and $\textbf{x}^T$ its transpose. For a complex number $z \in \mathbb{C}$, we denote by $\bar{z}$ its conjugate. Moreover, let $\|$ denote the concatenation operation, $\bm{I}_n$ the identity matrix of size $n$, and $\bm{0}_{k}$ a vector of $k$ zeros.

\subsection{Approximate Homomorphic Encryption}\label{prelims:encryption}
Homomorphic encryption is a particular type of encryption that enables computations to be executed directly on ciphertexts. 
The most recent and practical homomorphic schemes rely on the hardness of the Ring Learning with Errors (RLWE) problem which states that, given a polynomial ring $\mathcal{R}_q$, for a secret polynomial $s$, it is computationally hard for an adversary to distinguish between $(a, a \cdot s+e)$ and $(a, b)$, where $e$ is a short polynomial sampled from a noise distribution, and $a, b$ are polynomials uniformly sampled over $\mathcal{R}_q$.

Cheon \etal recently introduced the CKKS cryptosystem~\cite{cheon_homomorphic_2017} (improved in \cite{cheon2018RNS}), an efficient and versatile leveled homomorphic scheme for approximate arithmetic operations. An approximate homomorphic encryption scheme enables the execution of approximate additions and multiplications on ciphertexts without requiring decryption. 
It uses an isomorphism between complex vectors and the plaintext space $\mathcal{R}_q{=}\mathbb{Z}_{q}[X] /(X^N{+}1)$, where $q$ is a large modulus, and $N$ is a power-of-two integer. 
The decryption of a ciphertext yields the input plaintext in $\mathcal{R}_q$ with a small error. This small error can be seen as an approximation in fixed-point arithmetic.

In CKKS, given a ring isomorphism between $\mathbb{C}^{N/2}$ and $\mathbb{R}[X] /(X^N{+}1)$, a complex vector $\bm{z} {\in} \mathbb{C}^{N/2}$ can be encoded into a polynomial $m$ denoted by a vector $\textbf{m}$ of its coefficients $\{m_0, {\dots}, m_{N-1}\} {\in} \mathbb{R}^N$ as $\textbf{m}{=}\frac{1}{N}(\bar{\bm{U}}^T {\cdot} \bm{z} {+} \bm{U}^T {\cdot} \bm{\bar{z}})$, where $\bm{U}$ denotes the $(N/2){\times} N$ Vandermonde matrix generated by the $2N$-th root of unity $\zeta_j{=}e^{5^j\pi i/N}$. This transformation is extended to $\mathcal{R}_q$ by a quantization. Then, considering a maximum number of levels $L$, a ring modulus $q{=}\prod_{i=0}^{L-1} q_i$ is chosen with $\{q_i\}$ a set of number theoretic transform (NTT)-friendly primes such that $\forall i {\in} [0,L-1], \, q_i {=} 1 \mod 2N$.

Let $\chi_\text{err}, \chi_\text{enc}$, and $\chi_\text{key}$, be three sets of small distributions over $\mathcal{R}_q$. Then, for an encoded plaintext $m \in \mathcal{R}_q$, the scheme works as follows:

\descr{KeyGen($\lambda,N,L,q$)}: for a security parameter $\lambda$ and a number of levels $L$, generate $\bm{sk}{=}(1, s)$ with ${s \leftarrow \chi_\text{key}}$, $\bm{pk}{=}(b, a)$ with ${a {\leftarrow} \mathcal{R}_{q}}$, ${b{=}-a \cdot s + e \mod q}$, and ${e {\leftarrow} \chi_\text{err}}$. Additional keys which are useful for the homomorphic computations (i.e., rotation, evaluation keys, etc.) are denoted by $\bm{evk}$. We refer the reader to~\cite{cryptoeprint:2019:688} for further details.

\descr{Encryption($m, \bm{pk}$)}: for $r_0 \leftarrow \chi_\text{enc}$ and $e_0, e_1 \leftarrow \chi_\text{err}$, output $\bm{ct} {=} (ct_0,ct_1){=}r_0 \cdot \bm{pk} + (m+e_0, e_1) \mod \; q $.

\descr{Decryption($\bm{sk}, \bm{ct}$)}: Output $\hat{m} {=} \langle \bm{ct}, \bm{sk}\rangle \mod  q_l $, where $\langle \cdot, \cdot \rangle$ denotes the canonical scalar product in $\mathcal{R}_{q_l}$ and $l$ the current level of the ciphertext.

For brevity, we denote the above three operations as $\text{KeyGen}(\lambda,N,q)$, $\text{Enc}_{\bm{pk}}(m)$, and $\text{Dec}_{\bm{sk}}(\bm{ct})$, respectively. The scheme's parameters are chosen according to the security level required (see~\cite{HomomorphicEncryptionSecurityStandard}) to protect the inputs and \textit{privacy}.


\subsection{BDOP Commitment}\label{prelims:commit}

Baum \etal~\cite{baum_more_2018} proposed the BDOP commitment scheme, that enables us to prove in zero-knowledge certain properties of the committed values to a verifier. Based on lattices, this scheme also builds on a polynomial ring $\mathcal{R}_q {=} \mathbb{Z}_q/(X^N{+}1)$, with the notable exception that $q$ is a prime that satisfies $q {=} 2d{+}1\mod4d$, for some power-of-two $d$ smaller than $N$.

BDOP is based on the hardness assumption of the module Short Integer Solution (SIS) and module Learning with Error (LWE)~\cite{langlois_worst-case_2015} to ensure its \textit{binding} and \textit{hiding} properties. 
We refer the reader to~\cite{baum_more_2018} for more details. For a secret message vector $\bm{m}{\in} \mathcal{R}^{l_c}_{q}$, and for a commitment with parameters $(n,k)$, two public rectangular matrices $\bm{A}_1'$ and $\bm{A}_2'$, of size $n{\times}(k{-}n)$ and ${l_c}{\times}(k{-}n{-}{l_c})$ respectively, are created by uniformly sampling their coefficients from $\mathcal{R}_{q}$. To commit the message $\bm{m}$, we sample $\bm{r}_c{\leftarrow} \mathcal{S}_\beta^k$, where $\mathcal{S}_\beta^k$ is the set of elements in $\mathcal{R}_{q}$ with $l_{\infty}$-norm at most $\beta$ and bounded degree, and compute
\begingroup
\[ \text{BDOP}(\bm{m},\bm{r}_c) {=} \begin{pmatrix} c_1\\ \bm{c_2} \end{pmatrix} {=} \begin{pmatrix} \bm{A}_1 \\ \bm{A}_2 \end{pmatrix} \cdot \bm{r}_c + \begin{pmatrix} \bm{0}_n\\ \bm{m} \end{pmatrix},\]
\endgroup
with $\bm{A}_1  {=} [\bm{I}_n \| \bm{A}_1' ]$ and $\bm{A}_2 {=} [\bm{0}_{l_c \times n} \| \bm{I}_{l_c} \| \bm{A}_2' ]$.

The BDOP commitment scheme can be used, with a $\Sigma$-protocol, to provide a \textit{bound proof}: proof that a committed value is in a bounded range~\cite{baum_ecient_nodate}. The main rationale behind this is to prove in zero-knowledge that the committed value plus a small value has a small norm. 
Given a commitment $\bm{c}{=}\text{BDOP}(\bm{m},\bm{r}_c)$, the prover computes a commitment for a vector of small values $\bm{\mu}$ as  $\bm{t}{=}\text{BDOP}(\bm{\mu},\bm{\rho})$ and commits to this commitment in an auxiliary commitment $c_{\text{aux}}{=}C_{\text{aux}}(\bm{t})$. The verifier selects a challenge $d\in\{0,1\}$ and sends it to the prover who verifies its small norm and eventually opens $c_{\text{aux}}$. The prover also opens $\bm{t}+d \cdot \bm{c}$ to $\bm{z}{=}\bm{\mu}+d \cdot \bm{m}$ and $\bm{r}_z{=}\bm{\rho}+d \cdot \bm{r}_c$. Upon reception, the verifier checks that $\text{BDOP}(\bm{z},\bm{r}_z){=}\bm{t}+d \cdot \bm{c}$ and that the norms are small. The protocol, presented in Appendix~\ref{app:bdopBP}, is repeated to increase soundness and can be made non-interactive using the Fiat-Shamir heuristic.

\subsection{Zero-Knowledge Circuit Evaluation}\label{prelims:zkce}

Zero-knowledge circuit evaluation (ZKCE) protocols enable a user to prove the knowledge of an input that yields a public output on an arithmetic or Boolean circuit that implements a specific public function~\cite{chase_post-quantum_2017,giacomelli_zkboo_2016}. A circuit is defined as a series of gates connected by wires. Based on the multi-party computation (MPC) \textit{in-the-head} approach from Ishai \etal~\cite{ishai_zero-knowledge_2009}, ZKCE techniques emulate players and create a decomposition of the circuit (see Appendix~\ref{app:zkb}). The secret is shared among the emulated players, who evaluate the circuit in a MPC fashion and commit to their respective states. The prover then reveals the states of a subset of players depending on the verifier's challenge. By inspecting the revealed states, the verifier builds confidence in the prover's knowledge.



In particular, ZKB++~\cite{chase_post-quantum_2017} is a $\Sigma$-protocol for languages of the type $\{y \, | \, \exists x \; \text{s.t.} \\ y{=}\Phi(x)\}$, where $\Phi(\cdot)$ is the representation of the circuit. With randomized runs, the verifier builds confidence in the prover's knowledge of the secret. The number of iterations is determined according to the desired soundness: For instance, to prove the knowledge of a message that yields a specific SHA-256 digest, a security level of 128-bits requires 219 iterations. The proof size is linked to the number of iterations but also to the number of gates that require non-local computations (e.g., AND for Boolean circuits, multiplication for arithmetic ones). Compared to earlier work, i.e., ZKBoo~\cite{giacomelli_zkboo_2016}, ZKB++ reduces the proof size by not sending information that can be computed by the verifier.
The security of ZKB++ is based on the quantum random oracle model. Overall, it achieves the following properties: (a) \textit{2-privacy}, opening two out of the three players' views to the verifier reveals no information regarding the secret input, (b) \textit{soundness}, a correct execution yields a valid witness with soundness error linked to the number of iterations, and (c) \textit{completeness}, an honest execution of ZKB++ ensures a correct output.

\section{Architecture}\label{archi}

We now present our construction that enables computations on third-party certified data in a privacy and integrity preserving manner. It builds on (i) CKKS to encrypt the data and enable computations on it, and (ii) MPC-in-the-head and BDOP commitments to simultaneously verify a custom circuit that checks the integrity of the data and its correct encryption.
Its workflow is decomposed into five phases: \textit{collection}, \textit{transfer}, \textit{verification}, \textit{computation}, and \textit{release}. (1) In the collection phase, the user obtains data about herself or her activities from the data source, along with a certificate that vouches for its integrity and authenticity. (2) The user then encrypts the data, generates a proof for correct encryption of the certified data, and sends it with the ciphertexts to the service provider. (3) The service provider verifies the proof in the verification phase. Then, (4) it performs the desired computations on it, and (5) communicates with the user to obtain the corresponding result in the release phase.

\subsection{Collection Phase}\label{archi:CP}

In this phase, the user (identified by her unique identifier $uid$) collects from the data source certified data about herself or her activities. The data source certifies each user's data point $\bm{x}$ using a digital signature $\sigma(\cdot)$ that relies on a cryptographic hash function $H(\cdot)$ to ensure integrity. We opt for SHA-256 as the hash function due to its widespread use as an accepted standard for hash functions~\cite{fips_180-4}; our solution works with any signature scheme building on it. For example, Bernstein \etal~\cite{BernsteinSPHINCS19} recently proposed a quantum-secure signature scheme employing SHA-256. In more detail, the data source generates a payload $msg{=}\{nonce, uid, \bm{x}\}$ and sends to the user a message $M_0$ defined by: $M_0 {=} \{msg, \sigma(H(msg))\}$.

\begin{figure*}[t]
    \centering
    \includegraphics[width=\textwidth]{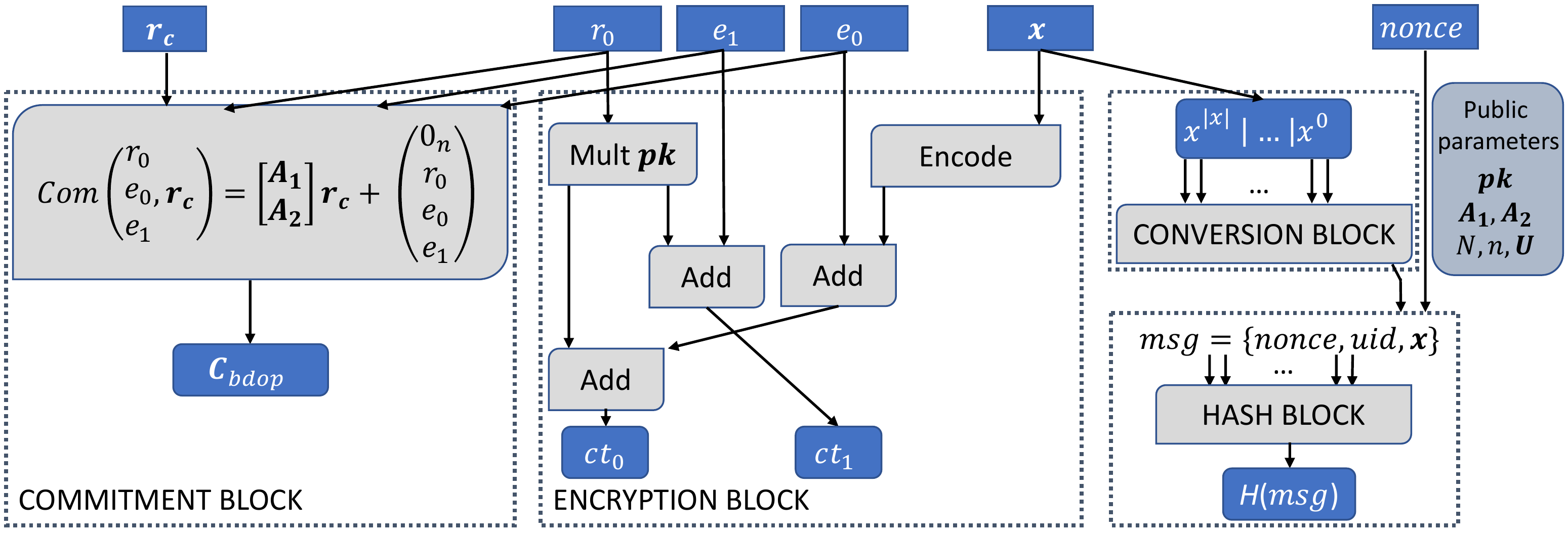}
    
    \caption{Overview of the verification circuit $\mathcal{C}$. Its inputs are denoted by rectangles and its outputs by rounded rectangles.}
    \label{fig:zkce}
    
\end{figure*}

\subsection{Transfer Phase}\label{archi:TP}

In this phase, the user protects her certified data points with the CKKS homomorphic encryption scheme (see Section~\ref{prelims:encryption}) and generates a proof of correct protection. To this end, \sysname employs a ZKCE approach to simultaneously prove the integrity of the underlying data and its correct encryption, i.e., to convince a service provider that the noises used for encryption did not distort the plaintexts. In particular, the user evaluates a tailored circuit $\mathcal{C}$ (depicted in Figure~\ref{fig:zkce}) that (i) computes the encryption of the data with the CKKS scheme, (ii) generates BDOP commitments to the noises used for encryption, and (iii) produces the hash digests of the messages signed by the data source to verify their integrity. 
For ease of presentation, we describe the circuit that processes one data point $\bm{x}$. However, this can easily be extended to a vector $\bm{d}$ obtained from multiple data points $\{\bm{x}_i\}$. The circuit's structure is publicly known and its public parameters are the encryption public information $\bm{pk}, \bm{U}, N$, the matrices $\bm{A}_1, \bm{A}_2$ used in the BDOP commitment scheme and its parameter $n$, and additional information such as the user's identifier.
The circuit's private inputs are the user's secret data point $\bm{x}$ and nonce, the encryption private parameters $r_0$, $e_0$, and $e_1$, and the private parameters of the BDOP commitment scheme $\bm{r}_c$. These inputs are arithmetically secret-shared among the three simulated players, according to the ZKB++ protocol. The outputs of the circuit are the ciphertext $\bm{ct}$, the commitment to the encryption noises $\bm{C}_{\text{bdop}}{=}\text{BDOP}({( r_0, e_0, e_1)^T},\bm{r}_c)$, and the digest of the message $H(msg)$ signed by the data source.

CKKS and BDOP operate on slightly different polynomial rings, as described in Section~\ref{prelims}. Consequently, we extend BDOP to the composite case where $q$ is a product of NTT-friendly primes. We relax the strong condition on the challenge space from~\cite{baum_more_2018} that \textit{all} small norm polynomials in $\mathcal{R}_q$ be invertible. This condition is required for additional zero knowledge proofs that are not used in our construction. We simply require that the challenge space of invertible elements be large enough to ensure the binding property of the commitment. In particular, considering that the divisors of zero in $\mathcal{R}_q$ are equally distributed in a ball $\mathcal{B}$ of norm $\beta_\text{c}$ as in $\mathcal{R}_q$, the probability of having a non-invertible element when uniformly sampling from $\mathcal{B}$ is at most $\frac{N\cdot L}{2^l}$, where $L$ is the number of prime factors in $q$, each having at least $l$ bits. As a result, the number of invertible elements in $\mathcal{B}$ is lower-bounded by ${|\mathcal{B}|*(1{-}\frac{N\cdot L}{2^l})}$, where $|\mathcal{B}|{=}(\beta_\text{c}+1)^N$ is the cardinality  of the ball. Thus, by adequately choosing $\beta_\text{c}$ and the product of primes, we create a sufficiently large challenge set of small-norm invertible elements in $\mathcal{R}_q$ (e.g., $>2^{256}$). Moreover, we note that our circuit requires computations to be executed on the underlying arithmetic ring $\mathbb{Z}_{q}$ used for the lattice-based encryption and commitment schemes, as well as a Boolean ring $\mathbb{Z}_{2}$ for the computation of the SHA-256 hash digests. We also design a block that converts MPC-in-the-head arithmetic shares of the input data of the circuit into Boolean ones.

Overall, our circuit $\mathcal{C}$ consists of four blocks, showed in Figure~\ref{fig:zkce}: encryption, commitment, conversion, and hash block.

\descr{Encryption Block.} This block operates in the arithmetic ring $\mathbb{Z}_{q}$ and takes as inputs the vector of integers in $\mathbb{Z}_{q}$ derived by quantization from the plaintext $\bm{x}$ produced during the data collection phase (see Section~\ref{archi:CP}), as well as the encryption with private noise parameters $r_0$, $e_0$, and $e_1$. It first encodes the secret input data to a polynomial $m \in \mathcal{R}_q$ before computing the ciphertext ${\bm{ct}{=}(ct_0, ct_1){=}r_0 \cdot \bm{pk} + (m + e_0, e_1) \mod q}$. This step requires only affine operations that can be computed locally for each simulated player of ZKB++ protocol. The encryption block is depicted in the middle part of Figure~\ref{fig:zkce}.

\descr{Commitment Block.} This block also operates in the arithmetic ring $\mathbb{Z}_{q}$; its inputs are the private parameters of the 
encryption (i.e., $r_0$, $e_0$, and $e_1$) and commitment (i.e., $\bm{r}_c$) schemes. As the commitment scheme has the same external structure as the encryption one, this block operates equivalently and returns $\text{BDOP}((r_0, e_0, e_1)^T,\bm{r}_c)$, requiring only local operations at each simulated player. An overview of the commitment block is shown in the leftmost part of Figure~\ref{fig:zkce}.

\descr{Conversion Block.} This block enables us to interface two types of circuits that would otherwise be incompatible when following a ZKCE approach. The main idea is to transform an arithmetic secret sharing into a Boolean secret sharing in the context of MPC-in-the-head. Let $[x]_B$ denote the Boolean sharing of a value $x$ and $[x]_A$ its arithmetic one.
An arithmetic additive secret sharing in $\mathbb{Z}_{q}$ splits $x$ into three sub-secrets $x_0$, $x_1$, and $x_2$ such that ${x {=} x_{0}{+}x_1{+}x_2\mod q}$. Let $x_i^k$, be the $k$-th bit of the arithmetic sharing of the secret $x$ for player $i$. A Boolean sharing $[x]_B$ cannot be directly translated from $[x]_A$ as the latter does not account for the carry when adding different bits. Considering that the modulus $q$ can be represented by $|q|$ bits, the conversion block generates $|q|$ Boolean sub-secrets $[y]^j_B{=}\{y^j_0, y^j_1, y^j_2\}_B$, such that
\begingroup
\[ \forall j\in [1, |q|] :\; x^{j} = \bigoplus_{i=0}^{2} y^{j}_i,\]
\endgroup
where $\oplus$ denotes the XOR operation (i.e., addition modulo~2), and $x^{j}$ is the $j$-th bit of $x$.
When designing such a block in the MPC-in-the-head context, we must make the circuit (2,3)-decomposable (see Appendix~\ref{app:zkb}) and ensure the 2-privacy property, i.e., revealing two out of the three players' views to the verifier should not leak any information about the input.

To reconstruct the secret in zero-knowledge and obtain a bit-wise secret sharing, the procedure is as follows: For every bit, starting from the least significant one, the conversion block computes (i) the sum of the bits held by each player, plus the carry from the previous bits, and (ii) the carry of the bit. The computation of the carry requires interaction between the different players (i.e., making the operation a ``multiplicative" one), hence we design a conversion block with a Boolean circuit that minimizes the amount of multiplicative gates.

More precisely, we design a bit-decomposition block for MPC-in-the-head building on Araki \etal's optimized conversion~\cite{araki_generalizing_2018} between a power-of-two arithmetic ring and a Boolean ring. Let $\text{Maj}(\cdot)$ be the function returning the majority bit among three elements. Then, the conversion circuit, for every bit $k \in [1, |x|]$, does the following:

\begin{enumerate}
\itemsep-0.25em 
\begingroup
\item locally reads $[\alpha_k]_B {=} \{x_0^{k}, x_1^{k}, x_2^{k}\}$ (i.e., for each player);
\item computes the first carry $[\beta_k]_B$ amongst those inputs:
\[\beta_k{=}\text{Maj}( x_0^{k}, x_1^{k}, x_2^{k}) {=} (x_0^{k} \oplus x_2^{k} \oplus 1)(x_1^{k} \oplus x_2^{k}){\oplus} x_1^{k} ;\]
\endgroup
\item  computes the second carry $[\gamma_k]_B$ among those inputs with $\gamma_0{=}\beta_0{=}0$:
\begingroup
\[\gamma_k{=}\text{Maj}( \alpha_k, \beta_{k-1}, \gamma_{k-1}) {=} (\alpha_k \oplus  \gamma_{k-1} \oplus 1)(\beta_{k-1} \oplus  \gamma_{k-1})\oplus \beta_{k-1} ;\]
\item sets the new Boolean sharing of the secret to 
\[ [y]^{k}_B {=} [\alpha_k] \oplus [\beta_{k-1}] \oplus [\gamma_{k-1}].\]
\endgroup
\end{enumerate}

\noindent To the best of our knowledge, this is the first time a bit-decomposition circuit is used for MPC-in-the-head, which enables to interface circuits working in different rings.

\descr{Hash Block.} This block uses the SHA-256 circuit presented in~\cite{giacomelli_zkboo_2016} to compute the hash digest of the message $msg{=}\{nonce, uid, \bm{x}\}$ signed by the data source in the collection phase.

\descr{Full Circuit.} With the above building blocks, and following the ZKB++ protocol, the user generates a proof that can convince the service provider that she has not tampered with the data obtained by the data source.

Furthermore, using BDOP's bound proof protocol (see Section~\ref{prelims:commit}) the user produces a proof of correct encryption, i.e., that the encryption noise has not distorted the underlying plaintext. The cryptographic material of the combined proofs (ZKCE \& BDOP) is denoted by $\mathcal{P}$. At the end of the transfer phase, the user sends to the service provider the message: $$M_1{=}\{\bm{ct}, \bm{C}_{\text{bdop}}, \mathcal{P}, H(msg), \sigma(H(msg))\}.$$ 


\subsection{Verification Phase}\label{archi:VP}
Upon reception of a message $M_1$, the service provider verifies the signature using the provided hash digest. If satisfied, it verifies the proof $\mathcal{P}$ by first evaluating the circuit $\mathcal{C}$ following the ZKB++ protocol and then checking the bound proof for the encryption noises. Hence, it is assured that $\bm{ct}$ is the encryption of a data point $\bm{x}$ giving the hash that has been certified by the data source.

\subsection{Computation Phase}\label{archi:ComP}

Using the homomorphic capabilities of the CKKS encryption scheme, the service provider can perform any operation with a bounded predefined multiplicative depth (and arbitrary depth, with bootstrapping~\cite{cheon_bootstrapping_2018}) on validated ciphertexts received by the user. In particular, CKKS enables the computation of a wide range of operations on ciphertexts: additions, scalar operations, multiplications, and a rescaling procedure that reduces the scale of the plaintexts. Those functions enable the computation of polynomial functions on the ciphertexts. Moreover, it supports the evaluation of other functions such as exponential, inverse or square root~\cite{cheon_homomorphic_2017,cheon_bootstrapping_2018, cheon_numerical_2019}, by employing polynomial approximations (e.g., least squares). Hence, the service provider can independently compute any number of operations on the user's encrypted data simply requiring interactions with the user to reveal their outputs (see Section~\ref{archi:RP}).

\subsection{Release Phase}\label{archi:RP}

\begin{figure}[t]
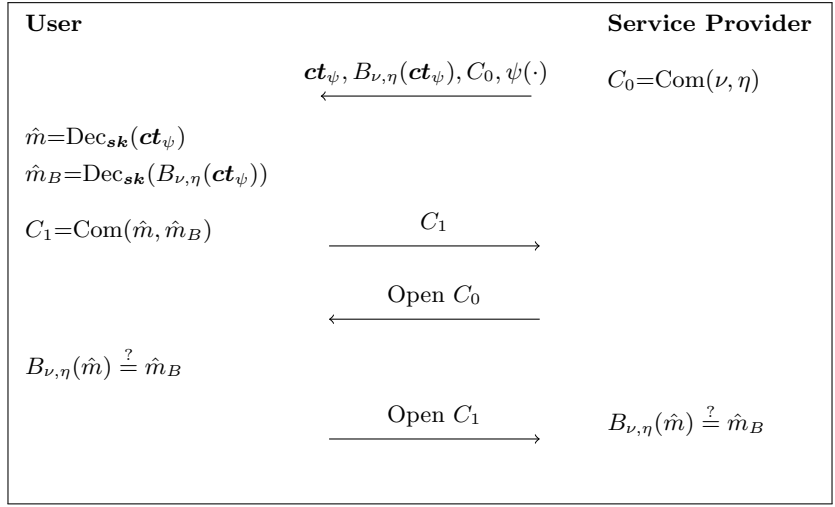

\begin{center}
\vspace{0.5em}
\fbox{
\pseudocode[colspace=1em]{
\textbf{User} \>\> \textbf{Service Provider}\\
\> \sendmessageleft*[2.8cm]{\bm{ct}_\psi, B_{\nu, \eta}(\bm{ct}_\psi), C_0, \psi(\cdot)} \> C_{0}{=}\text{Com}(\nu, \eta)\\
\hat{m}{=}\text{Dec}_{\bm{sk}}(\bm{ct}_\psi)\\
\hat{m}_B{=}\text{Dec}_{\bm{sk}}(B_{\nu,\eta}(\bm{ct}_\psi))\\
C_1{=}\text{Com}(\hat{m},\hat{m}_B) \> \sendmessageright*[2.8cm]{C_1} \>\\
\> \sendmessageleft*[2.8cm]{\text{Open $C_0$}}\>\\
B_{\nu,\eta}(\hat{m})\stackrel{?}{=} \hat{m}_B \>\>\\
\>\sendmessageright*[2.8cm]{\text{Open $C_1$}} \> B_{\nu,\eta}(\hat{m})\stackrel{?}{=} \hat{m}_B\\
}
}

\caption{Release protocol for a computed value $\hat{m}$.}
\label{fig:datarelease}
\end{center}
\end{figure}

At the end of the computation phase, the service provider holds a ciphertext of the desired output that can only be decrypted by the holder of the secret key. To this end, the service provider and the user engage in a two-round release protocol, which ensures the service provider that the decrypted output is the expected result of the computation on the user's data. The release protocol is depicted in Figure~\ref{fig:datarelease} and detailed next.

Let $\bm{ct}_\psi$ denote the ciphertext obtained by the service provider after performing computations on validated ciphertext(s), and $\hat{m}$ the corresponding plaintext. First, the service provider informs the user of the computation $\psi(\cdot)$ whose result it wants to obtain. Then, the service provider homomorphically blinds $\bm{ct}_\psi$ by applying the function $B_{\nu, \eta}(x) {=} \nu {\cdot} x {+} \eta$, with $\nu$ and $\eta$ uniformly sampled in $\mathbb{Z}_q^*$ and $\mathbb{Z}_q$ resp., and commits to the secret parameters used for blinding (i.e., $\nu, \eta$) using a hiding and binding cryptographic commitment $\text{Com}(\cdot)$ as $C_0{=}\text{Com}(\nu, \eta)$. A hash-based commitment scheme can be used for this purpose~\cite{chase_post-quantum_2017}. Subsequently, the service provider sends to the user the encrypted result $\bm{ct}_\psi$, its blinding $B_{\nu, \eta}(\bm{ct}_\psi)$, and the commitment $C_0$. Upon reception, the user checks if the function $\psi(\cdot)$ is admissible. If the user accepts the computation $\psi(\cdot)$, she decrypts both ciphertexts as: $\text{Dec}_{\bm{sk}}(\bm{ct}_\psi){=}\hat{m}$ and $\text{Dec}_{\bm{sk}}(B_{\nu, \eta}(\bm{ct}_\psi)){=} \hat{m}_B$. 
Then, she commits to the decrypted results, i.e., $C_1{=}\text{Com}(\hat{m}, \hat{m}_B)$, and communicates $C_1$ to the service provider who opens the commitment $C_0$ to the user (i.e., revealing $\nu, \eta$). The user verifies that the initial blinding was correct by checking if $B_{\nu, \eta}(\hat{m}) {\stackrel{?}{=}}\hat{m}_B$. If this is the case, she opens the commitment $C_1$ (i.e., revealing $\hat{m}, \hat{m}_B$) to the service provider who verifies that the cleartext result matches the blinded information (i.e., by also checking if $ {B_{\nu, \eta}(\hat{m})\stackrel{?}{=}}\hat{m}_B$). At the end of the release phase, both parties are confident that the decrypted output is the expected result of the computation, while the service provider learns only the computation's result and nothing else about the user's data.

\section{Privacy and Security Analysis}\label{analysis}

\sysname protects the user's privacy by revealing only the output of the agreed computation on her data, and it protects the service provider's integrity by preventing any cheating or forgery from the user. Here, we present these two properties and their corresponding proofs. The used lemmas and propositions are presented in Appendix~\ref{app:lem}.

\subsection{Privacy}\label{analysis:priva}

\begin{prop}\label{prop:priva}
Consider a series of messages $\{msg_i\}$ certified by the data source with a digital signature scheme $\sigma(\cdot)$ that uses a cryptographic hash function $H(\cdot)$ with nonces. Assume that the parameters of the CKKS $(N,q,\chi_\text{enc},\chi_\text{key},\\\chi_\text{err})$ and BDOP $(\beta,k,n,q,N)$ schemes have been configured to ensure post-quantum security, that the circuit $\mathcal{C}$ is a valid (2,3)-decomposition, and that the cryptographic commitment $\text{Com}(\cdot)$ is hiding and binding. Then, our solution achieves privacy by yielding nothing more than the result $\hat{m}$ of the computation on the user's data $\{\bm{x}_i\}$.
\end{prop}
\begin{proof}[Proof]
To prove the privacy of \syscomma, we construct an ideal simulator whose outputs are indistinguishable from the real outputs of \syscomma's transfer and release phases. 

\descr{Transfer Phase.} In the quantum random oracle model (QROM), consider an ideal-world simulator $\mathcal{S}_t$ and any corrupted probabilistic polynomial time (PPT) service provider (i.e., the verifier). Without loss of generality, we consider only one round of communication between the user and service provider (i.e., one set of challenges). The simulator $\mathcal{S}_t$ generates a public-private key pair $(\bm{pk}',\bm{sk}')$. Following the encryption protocol, $\mathcal{S}_t$ samples $r_0'\leftarrow\chi_{\text{enc}}$ and $e_0', e_1' \leftarrow \chi_{\text{err}}$ and computes the encryption of a random vector $\bm{m}'$ into $\bm{ct}'$. Similarly, it samples a commitment noise vector $\bm{r}'_c\leftarrow\mathcal{S}_\beta^{k}$ and commits $(r_0',e_0', e_1')$ into $\bm{C}'_\text{bdop}$. Using a random nonce, the simulator also hashes $H(\bm{m}'[0])$. Without loss of generality, this can be extended to all components of $\bm{m}'$. $\mathcal{S}_t$ then sends $\{\bm{ct}',\bm{C}'_\text{bdop},H(\bm{m}'[0])\}$ to the service provider. 
The view of the service provider in the real protocol comprises $\{\bm{ct},\bm{C}_\text{bdop},H(msg)\}$. By the semantic security of the underlying encryption scheme~\cite{cheon_homomorphic_2017}, the hiding property of the BDOP commitment scheme (see Lemma~\ref{lem:bdop4}), and the indistinguishability property of the hash function in the QROM, the simulated view is indistinguishable from the real view.

Following the proof of Lemma\;\ref{lem:bdop10} in \cite{baum_ecient_nodate}, for each iteration of the bound proof with challenge $\text{d}{\in} \{0,1\}$, the simulator $\mathcal{S}_t$ can randomly draw $z'$ and $\bm{r}_z'$ with small norm and set $\bm{t}{=}\text{BDOP}(z',\bm{r}_z'){-}\text{d}\bm{C}_\text{bdop}$ (see~\cite{baum_ecient_nodate}). The simulator then commits to $\bm{t}$ in the bound proof protocol. Both ideal and real distributions are indistinguishable by the hiding property of the auxiliary commitment. 

In parallel, following~\cite{giacomelli_zkboo_2016}, given $e \in \{1,2,3\}$, the simulator $\mathcal{S}_t$ sequentially

\begin{itemize}
    \item Evaluates the \textsc{Share} function on the vector $\bm{m}'$, the encryption noises $e_0'$, $e_1'$, and $r_0'$ and commitment noises $\bm{r}_c'$. We denote the result by $({\text{view}'}_1^{0},{\text{view}'}_2^{0},{\text{view}'}_3^{0})$ (See Appendix~\ref{app:zkb}).
    \item Samples random tapes $\bm{k}_e'$, $\bm{k}_{e+1}'$.
    \item Evaluates the arithmetic circuit according to:
    If gate $c$ is linear, it defines ${\text{view}'}_{e}^{c}$ and ${\text{view}'}_{e+1}^{c}$ using $\phi_e^c$ and $\phi_{e+1}^c$. If gate $c$ is a multiplication one, it samples uniformly at random ${\text{view}'}_{e+1}^{c}$ and uses $\phi_e^c$ to compute ${\text{view}'}_{e}^{c}$.
    \item Once all the gates are evaluated and the vectors of views $\textbf{View}_e'$ and $\textbf{View}_{e+1}'$ are defined (see App.~\ref{app:zkb}), the simulator computes the respective outputs $y'_e$ and $y'_{e+1}$. 
    \item Computes $y'_{e+2}=y-(y'_{e}+y'_{e+1})$.
    \item Computes $z'_e$ following step (ii) of the ZKB++ protocol using $\textbf{View}'_{e+1}$, $\bm{k}_e'$, $\bm{k}_{e+1}'$ (and optionally $x'_{e+2}$ depending on the challenge).
    \item Outputs ($z'_e$, $y'_{e+2}$).
\end{itemize}

The simulator $\mathcal{S}_t$ follows a protocol similar to the original ZKB++ protocol. The only difference is that for a multiplicative gate $c$, the simulated view value ${\text{view}'}_{e+1}^{c}$ is sampled uniformly at random, whereas the original view value $\text{view}_{e+1}^{c}$ is blinded by adding $R_i(c)-R_{i+1}(c)$, with $R_i(c)$ and $R_{i+1}(c)$ the outputs of a uniformly random function sampled using the tapes $\bm{k}_{e}$ and $\bm{k}_{e+1}$. Thus, the distribution of $\text{view}_{e+1}^{c}$ is uniform and ${\text{view}'}_{e+1}^{c}$ follows the same distribution in the simulation. Therefore, the ZKB++ simulator’s output has the same distribution as the original transcript ($z_e$, $y_{e+2}$)  the output of the simulator $\mathcal{S}_t$ is indistinguishable from the valid transcript to a corrupted verifier.
Following the ideal functionality of $\mathcal{S}_t$, the ideal view of the service provider (i.e., $\{\bm{ct}',\bm{C}'_\text{bdop},H(\bm{m}'[0]),P'\}$) is indistinguishable from the real view (i.e., $\{\bm{ct},\bm{C}_\text{bdop},H(msg),P\}$, with $P$ the real ZKB++ proof). Thus, the ideal and real outputs are indistinguishable for the corrupted PPT service provider proving the privacy-property of \syscomma's transfer phase.\qed

\descr{Release Phase.} We construct a second simulator $\mathcal{S}_r$ to prove that \syscomma's release protocol (Section~\ref{archi:RP}) reveals nothing more than the result $\hat{m}$ to a curious verifier. A different simulator is required, as the release phase is independent from the transfer phase. We consider that $\mathcal{S}_r$ knows the blinding function ahead of time (i.e., it knows $(\nu,\eta)$) for the real conversation leading to the service provider accepting $\hat{m}$. Upon reception of the first message $\{\bm{ct}_\psi, B_{\nu, \eta}(\bm{ct}_\psi), C_0, \psi(\cdot)\}$ such that $\text{Dec}_{\bm{sk}}(\bm{ct}_\psi)=\hat{m}$, $\mathcal{S}_r$ creates $\hat{m}_B$ using the blinding parameters. The simulator commits to $C_1'{=}\text{Com}(\hat{m},\hat{m}_B)$, which is indistinguishable from $C_1$ to the curious verifier according to the hiding property of the commitment scheme. After receiving an opening for 
$C_0$, the simulator opens $C_1'$ to $\hat{m}$ and $\hat{m}_B$, which sustain the verifier checks as defined in Section~\ref{archi:RP}. The binding property of the commitment scheme asserts that $(\nu,\eta)$ is used for the blinding. The aforementioned conversation between the prover and verifier is indistinguishable from the real conversation. By checking the function $\psi(\cdot)$, and as the service provider is honest-but-curious, the user is assured that the service provider evaluated $\psi(\cdot)$ and is not using her as a decryption oracle. If the user deems the function inadmissible, she aborts.
\end{proof}

\subsection{Integrity}\label{analysis:inte}
\begin{prop}\label{prop:inte}
Consider a series of messages $\{msg_i\}$ certified by the data source with a digital signature scheme $\sigma(\cdot)$ that uses a cryptographic hash function $H(\cdot)$ with nonces. Assume that the parameters of the CKKS $(N,q,\chi_\text{enc},\chi_\text{key},\\ \chi_\text{err})$ and BDOP $(\beta,k,n,q,N)$ schemes have been configured to ensure post-quantum security, that the ZKB++ protocol execution of $\mathcal{C}$ achieves soundness $\kappa$, that the blinding function $B_{\nu,\eta}$ is hiding, and that the cryptographic commitment $\text{Com}(\cdot)$ is hiding and binding. Then, our solution achieves integrity as defined in Section~\ref{model:objectives}, as it ensures with soundness $\kappa$ that the output $\hat{m}$ is the result of the computation on the user's data.
\end{prop}
\begin{proof}[Proof]
Let us consider a cheating user with post-quantum capabilities as defined in Section~\ref{model:threat}. She wants to cheat the service provider in obtaining from the public function $\psi(\cdot)$ a result that is not consistent with the certified data. The public function evaluated by the service provider is $\psi(\cdot)$ and returns $\hat{m}$ on the series $\{msg_i\}$ of data signed by the data source with the signature scheme $\sigma(\cdot)$. We interchangeably denote by $\psi(\cdot)$ the public function in the plaintext and ciphertext domains.
By Lemma\;\ref{lem:ckks}, the ciphertext $\bm{ct}_{\psi}$ can be decrypted correctly using the secret key $\bm{sk}$.
As stated in~\cite{giacomelli_zkboo_2016} adapted to~\cite{chase_post-quantum_2017}, the binding property of the commitments used during the MPC-in-the-head guarantees that the proof $\mathcal{P}$ contains the information required to reconstruct $\textbf{View}_e$ and $\textbf{View}_{e+1}$. Given three accepting transcripts (i.e., one for each challenge), the verifier can traverse the decomposition of the circuit from the outputs to the inputs, check every gate and reconstruct the input. By surjectivity of the ZKB++ decomposition function, the verifier can reconstruct $x'\text{ s.t. } \Phi(x')=y$ proving the 3-special soundness property (see proof of Proposition~\ref{lem:zkboo} in~\cite{giacomelli_zkboo_2016}). The completeness property of the ZKCE evaluation follows directly from the construction of the (2,3)-decomposition of the circuit. Thus, from a correct execution of $\tau$ iterations of the protocol (parameterized by the security parameter $\kappa$), a user attempting to cheat the ZKB++ execution will get caught by the service provider with probability at least $1{-}2^{-\kappa}$.
Hence, a malicious but rational user can only cheat by tampering with the data before they are input to the circuit, i.e., the input messages or the encryption parameters. As the user is rational, she samples proper noise for the BDOP commitment; otherwise, she would lose either privacy or utility: not sampling noise uniformly from $\mathcal{S}_\beta^k$ would lead to a privacy leakage; conversely, sampling noises in $\mathcal{R}_q$ with norm bigger than $\beta$ or with a degree above the threshold defined by the scheme would lead to an improperly formatted commitment, and thus, a potential loss in utility, as the service provider would reject it. By the collision-resistance property of the hash function, it is computationally infeasible for the user to find a collision and thus a tampered message yielding the same hash. By property~\ref{lem:bdop10}, the bound proof is correct and offers special soundness: the service provider will detect a cheating user that samples malicious noises to distort the encryption, with probability at least $1-2^\kappa$. Note that in the case of an abort, the protocol is simply re-executed. Finally, during the release protocol, the integrity of the computation's output $\hat{m}$ is protected by the hiding property of commitment $C_0$, the hiding property of the blinding function (seen as a one-time-pad shift cipher in $\mathbb{Z}_q$ which achieves perfect secrecy of $\nu{\cdot}x$, i.e., it is impossible for the user to find $\nu$ and blind another result as $\hat{m}_B$), and the binding property of $C_1$~\cite{giacomelli_zkboo_2016}. Therefore, in \sysname users can only cheat with probability at most $2^{-\kappa}$.
\end{proof}

\section{Evaluation}\label{eval}

We evaluate \sysname on the three use cases discussed in Section~\ref{intro}, namely smart metering, disease susceptibility, and location-based activity tracking, using public real-world datasets. We detail the instantiation and parameterization of our proposed solution, then illustrate its overall performance per use case, in terms of both overhead and utility.
As previously mentioned, \sysname enables to offload the data and to conduct multiple operations on it. For simplicity, we present only one operation per dataset.

\begin{figure*}[t]
\centering
\begin{subfigure}[]{0.33\textwidth}
    \includegraphics[width=\textwidth, trim= 20 17 20 20, clip]{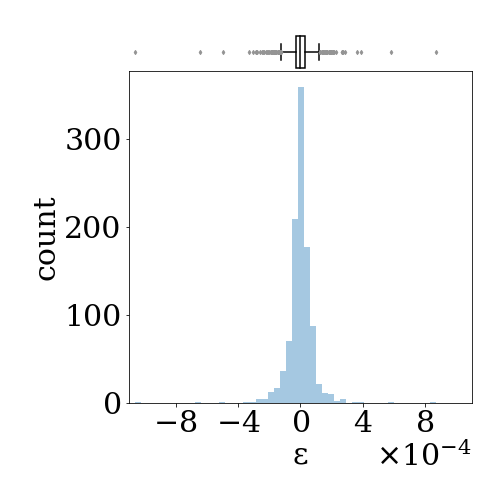}
    \vspace{-0.8cm}
    \caption{Smart Metering\\ \centering{(Addition)}}
    \label{fig:smhist}
\end{subfigure}
~
\begin{subfigure}[]{0.33\textwidth}
    \includegraphics[width=\textwidth, trim= 20 17 20 20, clip]{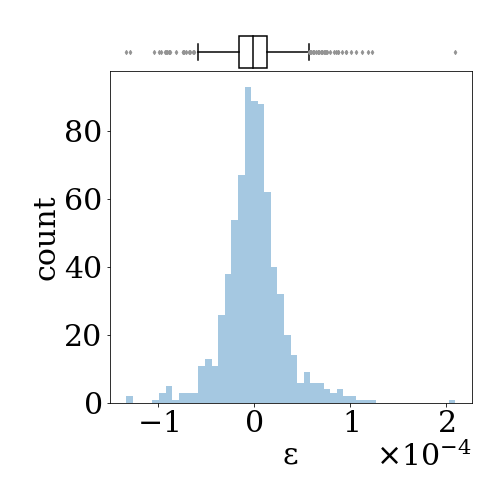}
    \vspace{-0.8cm}
    \caption{Disease Susceptibility \\ \centering{(Weighted Sum)}}
    \label{fig:genhist}
\end{subfigure}
~
\begin{subfigure}[]{0.33\textwidth}
    \includegraphics[width=\textwidth, trim= 20 17 20 20, clip]{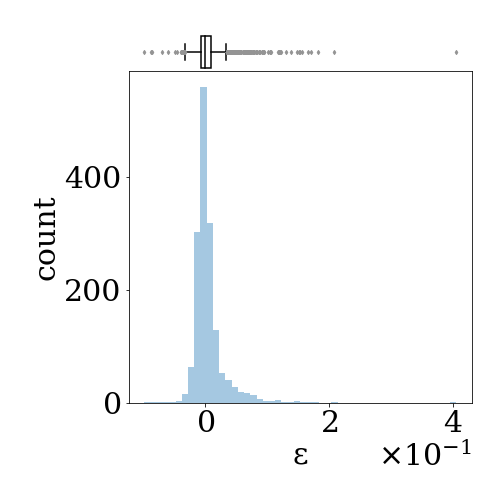}
    \vspace{-0.8cm}
    \caption{Location-based Activity-Tracking \\ \centering{(Euclidean Distance)}}
    \label{fig:runhist}
\end{subfigure}
\vspace{-0.35cm}
\caption{Histogram and boxplot of the relative error for the three use cases. 
The boxes shown on top of each figure represent the interquartile range (IQ) and the median, the whiskers are quartiles $\pm 1.5\cdot$IQ, and the dots are outliers.}
\label{fig:relative-error}

\end{figure*}

\subsection{Implementation Details}\label{eval:implem}
We detail how the various blocks of our construction are implemented and configured.

\descr{Implementation.} We implement the various blocks of \sysname on top of different libraries. The homomorphic computations are implemented using the Lattigo library~\cite{lattigo}. The commitment and encryption blocks of the circuit are implemented using CKKS from~\cite{snucrypto} by employing a ZKB++ approach. The circuit's Boolean part (i.e., the hash and conversion blocks) is implemented on top of the SHA-256 MPC-in-the-head circuit of~\cite{giacomelli_zkboo_2016}. All the experiments are executed on a modest Manjaro~4.19 virtual machine with an i5-8279U processor running at 2,4~GHz with 8\,GB~RAM.

\descr{CKKS \& BDOP.} For CKKS, we use a Gaussian noise distribution of standard deviation $3.2$, ternary keys with i.i.d. coefficients in $\{0,\pm 1\}^N$, and we choose $q$ and $N$ depending on the computation and precision required for each use case, such that the achieved bit security is always at least 128 bits. Each ciphertext encrypts a vector $\bm{d}$ consisting of the data points $\{x_i\}$ in the series of messages $\{msg_i\}$. Our three use cases need only computations over real numbers, hence we extend the real vector to a complex vector with null imaginary part. Similarly, the BDOP parameters for the commitment to the encryption noises are use case dependent. In principle, we choose the smallest parameters $n$ and $k$ to ensure a $128$-bit security ($n{=}1$, $k{=}5$) and $\beta$ is chosen according to $N$ and $q$.

\descr{ZKB++.} We set the security parameter $\kappa$ to 128, which corresponds to 219 iterations of the ZKB++ protocol. We also consider seeds of size $128$ bits and a commitment size of $|c|{=}256$ bits using SHA-256 as in~\cite{chase_post-quantum_2017}. Overall, considering the full circuit, the proof size per ZKB++ protocol iteration $|p_i|$ is calculated as
\begingroup
\begin{multline*}
|p_i| {=} |c| + 2\kappa + \log_23 +  \frac{2}{3}(|\bm{d}| + |\text{Com}| + |\text{Enc}| + |\bm{t}|) + b_{\text{hash}} + b_{\text{A2B}},
\end{multline*}
\endgroup
with $|\bm{d}|$ being the bit size of the secret inputs, $|\text{Com}|$ the bit size of the commitment parameters, $|\text{Enc}|$ the bit size of the encryption parameters, $b_{\text{hash}}$ the number of multiplicative gates in the SHA-256 circuit, $b_{\text{A2B}}$ the number of AND gates in the conversion block, and $|\bm{t}|$ the bit size of the additional information required to reconstruct the data source's message but not needed for the service provider's computation (e.g., user identifier, nonce, timestamps, etc.). We note that according to the NIST specification~\cite{fips_180-4}, SHA-256 operates by hashing data blocks of 447 bits. If the size of the user's input data exceeds this, it is split into chunks on which the SHA-256 digest is evaluated iteratively, taking as initial state the output of the previous chunk (see~\cite{fips_180-4}). We adapt the SHA-256 Boolean circuit described in~\cite{giacomelli_zkboo_2016}, which uses 22,272 multiplication gates per hash block, to the setting of ZKB++~\cite{chase_post-quantum_2017}. The Boolean part of the circuit is focused on the $|\bm{x}|$ least significant bits of the arithmetic sharing of $\bm{d}$ which is concatenated locally with a Boolean secret sharing of the additional information (nonce, uid, etc.). In our implementation, the user needs 182\,ms to run the Boolean part of the circuit associated with generating a hash from a 32-bits shared input $\bm{x}$. The verifier needs 73\,ms to verify this part of the circuit.

\descr{Release Protocol.} We use SHA-256 as a commitment scheme $\text{Com}(\cdot)$ and a linear blinding operation $B_{\nu,\eta}(\cdot)$ in $\mathbb{Z}_q$.

\descr{Evaluation Metrics.} We evaluate the performance of our solution on different use cases with varying complexity in terms of computation (i.e., execution time) and communication (i.e., proof size) overhead. The proof $\mathcal{P}$ is detailed as the proof for the ZKCE, as well as the BDOP bound proof. We also report the optimal ZKCE proof size per datapoint: i.e., if the ciphertexts are fully packed. To cover a wide range of applications we evaluate various types of operations on the protected data such as additions, weighted sums, as well as a polynomial approximation of the non-linear Euclidean distance computation. As CKKS enables approximate arithmetic, we measure the accuracy of our solution by using the relative error. Given the true output of a computation $m$ and the (approximate) value $\hat{m}$ computed with \syscomma, the relative error $\epsilon$ is defined as $\epsilon {=}\frac{m - \hat{m}}{m}$.

\subsection{Smart Metering}\label{eval:SM}
We consider a smart meter that monitors the household's electricity consumption and signs data points containing a fresh nonce, the household identifier, the timestamp, and its consumption. The energy authority is interested in estimating the total household consumption (i.e., the sum over the consumption data points) over a specified time period $I$ (e.g., a month or a year) for billing purposes\vspace{-0.35cm}

\begingroup
\setlength\abovedisplayskip{0em}
\setlength\belowdisplayskip{0em}
\[\small{m_\text{sm} = \sum_{i\in I} \bm{d}[i]},\]
\endgroup%
where $\bm{d}$ is the vector of the household consumption per half hour.
As our solution offloads the encrypted data to the service provider, additional computations, e.g., statistics about the household's consumption, are possible without requiring a new proof; this improves flexibility for the service provider.

\descr{Dataset \& Experiment Setup.} We use the publicly available and pre-processed UKPN dataset~\cite{ukpn_smartmeter_nodate} that contains the per half hour (phh) consumption of thousands of households in London between November 2011 and February 2014. Each entry in the dataset comprises a household identifier, a timestamp, and its consumption phh. 
For our experiment, we randomly sample a subset of 1,035 households and estimate their total energy consumption over the time span of the dataset with our solution. We set the parameters as follows: We use a modulus of $\log q {=} 45$ bits and a precision of 25 bits, which imposes a maximum of $2^{10}$ slots for the input vectors ($\log N {=} 11$). Hence, each household's consumption phh is encoded with multiple vectors $\bm{d}_k$ to cover the time span of the dataset. To evaluate its proof size, we assume that the messages obtained from the smart meter include a 16-bit household id, a 128-bit nonce, a 32-bit timestamp, and a 16-bit consumption entry.

\descr{Results.} The average time for encryption of a vector of 1,024 datapoints at the user side is $t_\text{enc}{=}70$\,ms, and the decryption requires $t_\text{dec}{=}0.7$\,ms. The mean time for the energy computation at the service provider side is $t_\text{comp}{=}130$\,ms. To generate the proof for one ciphertext, containing 1,024\,phh measurements (i.e., 21 days worth of data), the user requires $t_\text{prove}{=}3.3$\,min, and its verification at the service provider's side is executed in $t_\text{ver}{=}1.4$\,min. 
The estimated ZKCE proof size for each ciphertext of 1,024 elements is $643.4$\,MB, whereas the bound proof is $7.05$\,MB. For fully packed ciphertexts (1,024 datapoints), \syscomma's proof generation and verification respectively take $195$~ms and $80$~ms per datapoint, with a communication of $628$~KB.
Finally, Figure~\ref{fig:smhist} displays the accuracy results for the smart metering use case. We observe that our solution achieves an average absolute relative error of $5.1{\cdot}10^{-5}$ with a standard deviation of $7.2{\cdot}10^{-5}$, i.e., it provides very good accuracy for energy consumption computations. We remark that more than $75$\% of the households have an error less than $\pm 2.5{\cdot}10^{-4}$.

\begin{table}
\footnotesize
\centering
\setlength\tabcolsep{2pt}
\renewcommand{\arraystretch}{0.5}
\caption{Evaluation summary of \sysname (reported timings are the averages over 50 runs $\pm$ their standard deviation).}
\hspace*{-1.5cm}
\begin{tabular}{ L{2.5cm} C{1.7cm} C{1.7cm} C{1.6cm} C{1.5cm} C{1.1cm} C{1.1cm} C{1.1cm} C{1.1cm}  } 
\toprule
 Use Case & Computation & Mean Absolute Relative Error & \normalsize{$t_\text{enc}$} \footnotesize{(ms)} & \normalsize{$t_\text{comp}$} \footnotesize{(ms)}& \normalsize{$t_\text{dec}$} \footnotesize{(ms)}& Proof Size~(MB) & \normalsize{$t_\text{prove}$} \footnotesize{(s)}& \normalsize{$t_\text{ver}$} \footnotesize{(s)}\\ 
\midrule
\textbf{Smart Metering} & Sum & $5.1\cdot10^{-5}$ & $70\pm 10$ & $130\pm30$ & $0.7\pm0.3$ & $650.5$ & $200\pm10$ & $82\pm5$\\ 
\midrule
\textbf{Disease Susceptibility} & Weighted Sum &  $2.2\cdot 10^{-5}$ & $60\pm10$ &$22\pm5$ & $2.7\pm0.8$ & $53.9$ & $26\pm4$ & $13\pm2$\\ 
\midrule
\textbf{Location-Based Activity Tracking} & Euclidean Distance & $1.5\cdot10^{-2}$ & $980\pm70$ & $180\pm30$ & $7\pm2$ & $1,603$ & $470\pm40$ & $210\pm10$\\
\bottomrule
\end{tabular}
\label{table:perf}

\end{table}

\vspace{-0.32cm}
\subsection{Disease Susceptibility}\label{eval:DS}
We assume a medical center that sequences a patient's genome and certifies batches of single nucleotide polymorphisms (SNPs) that are associated with a particular disease $\partial$. A privacy conscious direct-to-consumer service is interested in estimating the user's susceptibility to that disease by calculating the following normalized weighted sum
\begingroup
\setlength\abovedisplayskip{0.1cm}
\setlength\belowdisplayskip{0cm}
\[\small{m_\partial  = \sum_{i \in S_\partial} \omega_i \cdot \bm{d}[i]},\]
\endgroup
where $S_\partial$ is the set of SNPs associated with $\partial$ and $\omega_i$ are their corresponding weights. The vector $\bm{d}$ comprises of values in $\{0,1,2\}$ indicating the presence of a SNP in 0, 1, or both chromosomes, which can be represented by two bits. This use case illustrates the need for flexibility in the service provider's computations, since it may be required to evaluate several diseases on the same input data at different times. Moreover, it accentuates the need for resistance against quantum adversaries, since genomic data is both immutable and highly sensitive over generations.

\descr{Dataset \& Experiment Setup.} We employ the 1,000 Genomes public dataset \cite{noauthor_1000_nodate}, that contains the genomic sequences of a few thousands of individuals from various populations. We randomly sample 145 individuals and extract 869 SNPs related to five diseases: Alzheimer's, bipolar disorder, breast cancer, type-2 diabetes, and Schizophrenia. We obtain the weight of a SNP with respect to those diseases from the GWAS Catalog~\cite{gwas}. Then, for every individual, we estimate their susceptibility to each disease. For this use case, we use a precision $\log p {=} 25$, a modulus of $\log q {=} 56$ consumed over two levels and a polynomial degree of $\log N {=} 12$. The input vector $\bm{d}$ (consisting of $2^{11}$ slots) is an ordered vector of integers containing the SNP values, coded on two bits, associated with the diseases. One vector is sufficient for the considered diseases. To estimate the proof size, we assume that the message signed by the data source contains a 16-bit user identifier, a 128-bit nonce, and the whole block of SNPs.

\descr{Results.} The average encryption time for up to 2,048 SNPs at the user side is $t_\text{enc}{=}60$\,ms, and the decryption is $t_\text{dec}{=}2.7$\,ms. The computation time of the disease susceptibility at the service provider is $t_\text{comp}{=}22$\,ms. The user needs $t_\text{prove}{=}26$\,s to generate the proof for the arithmetic part of the circuit, and the service provider verifies it in $t_\text{ver}{=}13$\,s. The estimated proof size for the ZKCE is $36.6$\,MB, whereas the bound proof is $17.3$\,MB. Figure~\ref{fig:genhist} shows our construction's accuracy for disease susceptibility computations by plotting the distribution of the relative error. We remark that the mean absolute relative error for such computations is appreciably low: $2.2{\cdot} 10^{-5}$ on average with a standard deviation of $2.3{\cdot}10^{-5}$. Moreover, more that $75$\% of the evaluated records have an absolute error inside the range $\pm0.7{\cdot}10^{-4}$.

\subsection{Location-Based Activity Tracking}\label{eval:LT}
We assume that a user is running with a wearable device that retrieves her location points during the activity from a data source, e.g., a cellular network. The service provider, e.g., an online fitness social network, seeks to estimate the total distance that the user ran during her activity $I$:
\begingroup
\setlength\abovedisplayskip{0em}
\setlength\belowdisplayskip{0em}
\[\small{m_\text{run} {=} \sum_{i \in I} \sqrt{(\bm{d}[i{+}1]{-}\bm{d}[i])^2 {+} (\bm{d}[\frac{N}{4}{+}i{+}1]{-}\bm{d}[\frac{N}{4}{+}i])^2}},\]
\endgroup
with $\bm{d}$ the vector of UTM (Universal Transverse Mercator) inputs packing Eastings in the first half of the vector and Northings in the second. Given that Euclidean distance computations require the evaluation of a non-linear square root function, we consider its least-squares approximation by a degree seven polynomial on a Legendre polynomial base.

\descr{Dataset \& Experiment Setup.} We run our experiment on a public dataset from Garmin Connect~\cite{international_garmin_nodate}. This dataset contains GPS traces of thousands of users engaging in various activities such as walking, running, and cycling. We randomly sample 
2,000 running traces
and we discard traces with less than 15 points and more than 2,000 points. Our initial dataset analysis shows that the traces are very \textit{noisy}: we identified unrealistic distances between consecutive points, timestamps and locations. We use GPSBabel~\cite{gpsbabel}, an open-source software, to interpolate the running traces such that the following criteria are met: (a) the maximum speed of a runner is less than $10$\,m/s, (b) the maximum distance between consecutive points is less than $30$\,m, and (c) the time delta between two points is less than $3$\,s, which 
are realistic for running activities. We remove traces whose time sampling was improperly executed by the data source (difference more than $10$\,s, standard deviation more than $5$, or a zero inter-quartile at $75$\%), as well as traces with unacceptable idleness\footnote{Idleness of a trace is a situation where the interquartile at $25$\% of the instant speed is less than $0.3$\,m/s and the covered distance is less than $15$\;m.}, and we convert the remaining GPS traces to UTM to obtain the Northings and Eastings geographic coordinates. Overall, we obtain a dataset of 1,608 traces (80\% of the initial 2K running trace dataset) which on average contain 1,124 datapoints and we estimate their total distance with \syscomma.

Considering the polynomial approximation required for the square root function, we set the size of the polynomial ring $N{=}2^{13}$ and a modulus $\log q {=} 184$. To calculate the proof sizes, we assume that the messages obtained from the data source contain a 16-bit user identifier, a 128-bit nonce, a 32-bit timestamp, and 24-bit Easting/Northing coordinates.

\descr{Results.} The encryption and decryption overhead for fully packed ciphertexts of up to 2,048 points at the user side is $t_\text{enc}{=}980$\,ms and $t_\text{dec}{=}7$\,ms, respectively, and the Euclidean distance computation at the service provider requires $t_\text{comp}{=}180$\,ms. For 2,048 datapoints, the user generates the proof for the arithmetic part of the circuit in $t_\text{prove}{=}7.9$\,min, and the service provider verifies it in $t_\text{ver}{=}3.4$\,min. Considering that each message signed by the data source is 96-bits, the proof size per trace for the ZKCE is $1,499.2$\,MB, and the bound proof is $103.7$\,MB. For our dataset, the average proof size is $922.1$~MB considering the mean number of points in the traces. In Section~\ref{eval:optz}, we will show how to reduce this proof size. With fully packed ciphertexts, \syscomma's proof generation requires $230$~ms per datapoint and $100$~ms for its verification, at a communication cost of $732$~KB.
Finally, Figure~\ref{fig:runhist} plots the relative error that we achieve for Euclidean distance computations. In particular, the average absolute relative error is $1.5{\cdot}10^{-2}$ with a standard deviation of $2.3{\cdot}10^{-2}$. In Figure~\ref{fig:runhist}, we see that more than $75$\% of the evaluated traces have an absolute error between $\pm 0.04$. We observe that the polynomial approximation of the square-root introduces errors higher than the other use cases. An improved accuracy can be achieved by increasing the polynomial degree, but this would require to increase upfront the encryption parameters ($N$, $L$, $q$) introducing additional communication and computation overhead to \syscomma. The wider spread of the relative error is due to the variance of the datapoints. Indeed, our analysis shows that gait, time sampling, and skewness of the speed distribution are among the factors that influence the overall relative error of the computations.

\subsection{Reducing the Communication Overhead}\label{eval:optz}

Table~\ref{table:perf} summarizes \syscomma's overhead for three use cases: smart metering, disease susceptibility, and location-based activity tracking. We observe that it introduces acceptable computational overhead at the user and service provider sides and that it achieves average absolute relative error between $2.2{\cdot}10^{-5}$ and $0.015$ for the desired computations. We remark however that our construction uses post-quantum secure lattice-based cryptographic primitives, such as encryption and commitment, and the MPC-in-the-head approach, to ensure the integrity of the user's data transfer. These come at the price of an increased communication (i.e., proof size). Therefore, we propose several improvements that one could employ to reduce this overhead and illustrate them in Figure~\ref{fig:optim} for the smart metering and location-based activity tracking use cases.

\descr{Random Integrity Checks (RIC).} A first optimization is to reduce the number of data points whose integrity is checked by the service provider. This introduces a trade-off between \syscomma's security level and its communication overhead. In particular, a service provider can decide to check only a subset of the input data hashes in the data verification phase, as we assume \textit{malicious but rational} users (Section~\ref{model:threat}) who will not cheat if there is a significant probability of getting caught. This is achieved through a sigma-protocol, that can be made non-interactive with the Fiat-Shamir heuristic: The user sends the ciphertext that encrypts all the datapoints (this can be seen as a commitment). Then, the service provider challenges a subset of datapoints to be hashed in the verification circuit.
Such a strategy enables a service provider to tune the solution depending on the level of confidence it has in the user. In Figure~\ref{fig:optim} we observe how the proof size decreases as the service provider checks fewer data blocks. For instance, if the service provider checks $20$\% of the data blocks in the verification phase (RIC-20\%), the proof size for location-based activity tracking drops from $1,499.2$\,MB to $497$\,MB (i.e., $243$~KB/datapoint), whereas for smart metering it decreases from $643.4$\,MB to $142.2$\,MB (i.e., $139$~KB/datapoint). This yields a reduction of more than $66$\% in the total ZKCE communication overhead. Computation times to generate and verify the proofs are also more than halved.

\begin{figure}
    \centering
    \includegraphics[width=.7\textwidth]{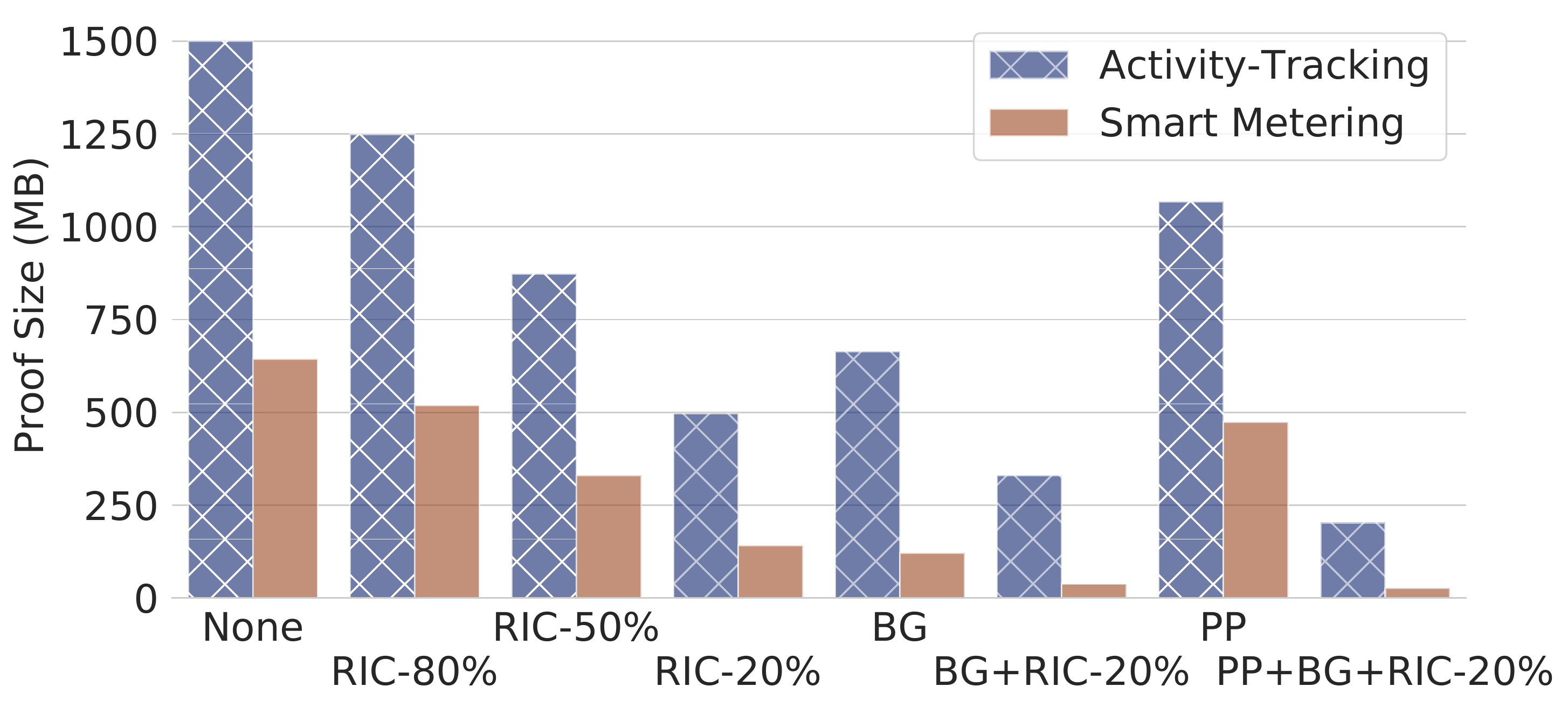}
    
    \caption{ZKCE proof size (MB) for fully packed ciphertexts and various optimizations.}
    \label{fig:optim}
    
\end{figure}

\descr{Batching (BG).} Another improvement is to modify the way data sources certify the users' data points. So far, in the smart metering and location-based activity tracking use cases, we assumed that data sources hash and sign every data point generated by the user. However, another strategy is to hash batches of data points in a single signed message. This modification is purely operational as it does not require additional software or hardware deployment. We set the batch size depending on the use case -- i.e., considering the additional information of each message before signature -- such that the overall batch can fit on a single SHA-256 input block of 447\,bits. Figure~\ref{fig:optim} shows a reduction of more than $50\%$ in proof size for the two use cases when batching (BG) is employed compared to the non-optimized solution. Batching can also be combined with RIC-20\% (BG+RIC-20\% in Figure~\ref{fig:optim}): For smart metering, the ZKCE proof size is further reduced to $38.1$\,MB (i.e., $37.2$~KB/datapoint), whereas for the location-based activity tracking the proof size drops to $329.7$\,MB (i.e., $161$~KB/datapoint). For activity-tracking, the $t_{\text{prove}}$ is reduced to $2.1$~min and $t_{\text{ver}}$ to $1.1$~min ($61$ and $32$~ms per datapoint, resp.). For smart metering, $t_{\text{prove}}$ is reduced to $20$~s and $t_{\text{ver}}$  to $9.3$~s ($20$ and $9$~ms per datapoint, resp.).

\descr{ZKCE Pre-processing (PP).} Finally, one can employ a ZKCE pre-processing model, such as that presented by Katz \etal~\cite{katz_improved_2018}. The pre-processing model considers that the user executes \textit{offline} a series of circuit evaluations on committed values. The service provider challenges a subset $\mathcal{M}$ of those evaluations and checks their integrity, and the remaining $\mathcal{\tau}$ ones are used in an \textit{online} phase along with the committed values. The rest of the protocol is similar to ZKB++. 
The proof size per iteration is reduced to:
\begingroup
\setlength\abovedisplayskip{0em}
\setlength\belowdisplayskip{-0.5cm}
\begin{small}
\begin{multline*}
|p_i| {=} 2\kappa + \mathcal{\tau} \log_2 \frac{\mathcal{M}}{\mathcal{\tau}} 3\kappa +\\[-0.3em]  \mathcal{\tau}(\kappa\log_2 3 + 2\kappa + (|\bm{d}| + |\text{Com}| + |\text{Enc}| + |\bm{t}|) + 2(b_{\text{hash}} + b_{\text{A2B}})).
\end{multline*}
\end{small}
\endgroup
Regarding our three players setting, a 128-bit security level requires $\mathcal{M}{=}300$ and $\mathcal{\tau}{=}81$, yielding a significant reduction of $25\%$ on the proof size (see~\cite{katz_improved_2018} for the computation details) compared to the non-optimized approach. 
Pre-processing, batching, and RIC can also be applied together to obtain smaller proofs (see PP+BG+RIC-20\% in Figure~\ref{fig:optim}): For smart metering, the ZKCE proof is reduced to $26.8$\,MB. Similarly, for location-based activity tracking, the ZKCE proof becomes $203.0$\,MB. This yields optimal ZKCE proof size per datapoint of $26.2$~KB and $99.1$~KB for smart metering and activity-tracking, respectively. Finally, we remark that according to Katz \etal~\cite{katz_improved_2018}, a trade-off between proof size and prover's computations could be achieved by increasing the number of players involved in the MPC-in-the-head protocol. However, such an improvement would require additional changes in \syscomma, e.g., the conversion block that interfaces the arithmetic and Boolean parts of the circuit should be adapted for a larger number of players.

\subsection{Comparison with ADSNARK} 
A fair comparison with~\cite{backes_adsnark:_2015} is not trivial to achieve, as our solution provides post-quantum security and overcomes the constraint of a trusted setup. Nonetheless, here we provide hints of their qualitative and quantitative differences. In particular, ADSNARK considers a smart metering use case that requires a non-linear cumulative function for the billing analysis over a month of data. We consider a similar non-linear pricing function evaluated by a degree-two polynomial, and we evaluate \sysname on the UKPN dataset for 400 households, with $N{=}12$ and $\log q {=} 106$. The median accuracy of our solution is higher than $99\%$. In terms of proof size, our construction yields $889.2$~MB (verifying all phh measurements for a month), whereas the overhead induced by ADSNARK is $71$~MB. However, we remark that the latter requires a new proof to be generated and exchanged every time a different computation is needed. In our solution, this cost is incurred only once; any subsequent operations can be computed locally by the service provider on the verified data. 
Additionally, ADSNARK accounts for only a ``\textit{theoretical estimate}'' of the complexity of the signature circuit (with only 1K multiplicative gates for signature verification) and, if we were to evaluate our solution with this circuit, the proof size would be only $104.2$\,MB. Thus, our analysis shows that our construction offers comparable results to the state-of-the-art and provides stronger security guarantees.

\section{Discussion}\label{discussion}
In this section, we present some interesting considerations that could influence the deployment of our solution.

\descr{Signature Scheme.} As discussed in Section~\ref{archi:CP}, \sysname is agnostic of the digital signature and it is compatible with any scheme that uses the  SHA-256 hash function. We employ SHA-256 as it is widely deployed in current infrastructures, adopted by various signature schemes (e.g., the recent post-quantum SPHINCS~\cite{BernsteinSPHINCS19} or the standard ECDSA schemes), and it is a benchmark for the evaluation of ZKBoo~\cite{giacomelli_zkboo_2016} and~ZKB++~\cite{chase_post-quantum_2017}.
This flexibility enables \sysname to be compliant with currently deployed signature schemes that might not be quantum resistant (e.g., ECDSA) at the cost of \syscomma's post-quantum integrity property.
Working with other hash functions (e.g., SHA-3 that is employed in~\cite{chase_post-quantum_2017}) is possible, with modifications to \syscomma's circuit.

\descr{Integrity Attacks.} \sysname copes with malicious users that might attempt to modify their data or the computed result to their benefit. However, some use cases require accounting for additional threats. For example, for smart metering, users might purposefully \textit{fail} to report some data (i.e., misreport) to reduce their billing costs. Similarly, in location-based activity-tracking, users might re-use pieces of data certified by the data source to claim higher performance and increase their benefits (i.e., double report). Such attacks can be thwarted by system level decisions; e.g., data sources can generate data points at fixed time-intervals known to service providers. Message timestamps can be encrypted along with the data points, so that service providers can verify their properties (e.g., their order or their range). As those attacks are application specific, we consider them out of the scope of this work.


\descr{Security of the ZKCE.} Dinur and Nadler~\cite{dinur_multi-target_2019} unveil a vulnerability of ZKCE systems, such as ZKB++, to multi-target attacks on the pseudo-random number generators. However, as stated by the authors, these attacks require a very large number of protocol executions (more than $2^{57}$) and thus are impractical and out of scope for our construction. The authors also argue that the use of appropriate salting in the pseudo-random number generation renders the attack very hard to succeed.
Additionally, Seker \etal propose a side-channel attack for MPC-in-the-head systems~\cite{seker2020sni}. However, we consider this attack out of scope of our model. 

\descr{Multiple Users.} \sysname trivially allows the service provider to compute aggregate statistics on data from multiple users by interacting separately with each of them and combining the results. Nonetheless, such a functionality can also be achieved by incorporating CKKS extensions to multiple parties. For instance, the multi-key scheme by Chen \etal~\cite{chen2019efficient} allows computations on ciphertexts generated by multiple users with their own keys. Alternatively, in the multiparty scheme by Mouchet \etal~\cite{mouchet2020multiparty} users generate a common public key for which the corresponding private key is secret shared among them. In both cases, the service provider computes on the users' ciphertexts and interacts with all of them to decrypt the result. 

\descr{Usability.} Even though \sysname introduces non-negligible communication and computation overhead, it remains acceptable for modern systems. The independent iterations of the ZKCE make the proof generation highly paralellizable and require much less memory than the full proof size (experimentally, as little as 2 GB of RAM). \sysname has also the advantage of being an offline system that only requires interaction in the release protocol: e.g., the transfer phase can be executed when the user is idle. Additionally, recent communication systems such as fiber optic internet or 5G offer high throughput links: With a 80\,Mb/s link, the proof for three weeks worth of smart metering data would only require about a minute to be transferred. For the activity tracking, \sysname can be executed when the user plugs her wearable device to a computer and transfer the data while recharging it.

\section{Related Work}\label{related}

Although homomorphic encryption is a solution receiving much traction to protect privacy in various fields, such as machine learning~\cite{graepel_ml_2013, juvekar_gazelle:_2018,shafagh_secure_2017} and medical research~\cite{kim_logistic_2018,wang_healer:_2016}, it only addresses the tension between privacy and utility, and it does not account for the authenticity of the encrypted data nor the correctness of the computation of its encryption (i.e., integrity). Verifiable encryption (VE) enables us to efficiently prove properties on encrypted data. Although VE solutions have been widely explored in the general case, notably by Camenish \etal~\cite{camenisch_verifiable_2000,camenisch_practical_2003}, they are still under investigation for lattice-based cryptographic systems that provide post-quantum security. Lyubashevsky and Neven~\cite{lyubashevsky_one-shot_2017} propose a \textit{one-shot} verifiable encryption for short solutions to linear relations, i.e., a single run of their protocol convinces a verifier that the plaintext satisfies the relation. Recent improvements, e.g.,~\cite{baum_more_2018} and~\cite{esgin_lattice-based_2019}, expand lattice-based VE to non-linear polynomial relations. Although VE can be used for proofs of correct encryption, it does not address data authenticity, which is ensured by cryptographic techniques, such as hash functions, that are more complex than polynomial relations.

To this end, homomorphic signatures~\cite{boneh_homomorphic_2011,catalano_practical_2013,catalano_security_2018,gorbunov_leveled_2015} and homomorphic authenticators~\cite{ahn_computing_2015,cheon_multi-key_2016,gennaro_fully_2013,matsui_context_2019} enable privacy-preserving computations on authenticated data. In particular, such schemes produce a signature of the plaintext result of homomorphic computations without deciphering. In this setting, a data owner provides a signature to some protected data and sends it to a server for processing. The server generates a new valid signature for the result of the homomorphic computation, which yields nothing more than the message it is signing. Some works, e.g.,~\cite{ahn_computing_2015}, improve this area with constructions offering homomorphic signatures that cope with low-degree polynomial operations. Homomorphic authenticators could be a solution to the problem under investigation, but they (a) require data sources to employ non-widely-supported homomorphic signature schemes, thus violating the minimal infrastructure modification requirement of our model, and (b) do not support data offloading at the service provider to amortize communication and storage costs.

Verifiable computation (VC)~\cite{fiore_efficiently_2014,gennaro_non-interactive_2010,lai_verifiable_2014} typically applies to cases where a computationally \textit{weak} user transfers her encrypted data to a cloud provider that computes on it, and its objective is to ensure the correctness and trustworthiness of the result. As such, VC protects only the integrity of the cloud computations and not the authenticity of the user's provided data. Such techniques are orthogonal to \sysname and could be employed to enable users to verify the computations performed by service providers, if the latter are considered malicious (and not honest-but-curious as in our Threat Model -- see Section~\ref{model:threat}).

Other VC techniques, known as zero-knowledge arguments, enable us to prove general statements about user private inputs. Pinocchio\,\cite{parno2013pinocchio} and subsequent works on succinct non-interactive arguments of knowledge (SNARKS), e.g.,~\cite{costello_geppetto:_2015}, build on Quadratic Arithmetic Programs \cite{gennaro2013quadratic} and bilinear maps to provide efficient proofs with small verification complexity. 
Backes \etal~\cite{backes_adsnark:_2015} extend SNARKS to the case of certified data (ADSNARK) and apply them to the three party model considered in this work. Similarly, ZQL~\cite{fournet_zql:_nodate} and $Z\emptyset$~\cite{fredrikson2014zo} present languages and compilers for data certification, client side computation, and result verification. But those solutions require computations from the user every time a new query is performed, thus not supporting data offloading. 
Furthermore, as pointed out by Katz \etal~\cite{katz_improved_2018}, SNARKS suffer from one or both of the following problems: (a) they require a trusted set-up, and (b) they are insecure against quantum attacks (due to the use of bilinear maps\;\cite{parno2013pinocchio,costello_geppetto:_2015,backes_adsnark:_2015,fournet_zql:_nodate,fredrikson2014zo}). 
Interestingly, this holds also for the recent work of Gennaro \etal~\cite{gennaro_lattice-based_2018}; they use LWE homomorphic encryption to achieve post-quantum security of their encodings, but still rely on trusted setups through common reference strings~\cite{damgard_efficient_2000} and q-PKE~\cite{PKEgroth2010short} assumptions. This trusted setup constraint is addressed by Wahby \etal~\cite{wahby_doubly-efficient_2018}, but their solution relies on the discrete log problem that is \textit{de facto} not post-quantum secure. Finally, Ben-Sasson \etal~\cite{ben-sasson_scalable_2019} propose STARKs which achieve transparent and scalable arguments of knowledge by relying only on the collision-resistant hash and Fiat-Shamir heuristic assumptions. Although STARK-like systems solve a similar problem to ours, they follow a different approach where the bulk of the work is executed at the user side and no data offloading is considered, hence, multiple data computations require the creation of multiple proofs.

A radically different approach to proving statements in zero-knowledge comprises solutions based on multi-party computation (MPC) such as ZKBoo~\cite{giacomelli_zkboo_2016}, ZKB++~\cite{chase_post-quantum_2017}, Ligero~\cite{ames_ligero:_2017}, and KKW~\cite{katz_improved_2018}. These solutions are built on top of the MPC-in-the-head paradigm introduced by Ishai \etal~\cite{ishai_zero-knowledge_2009} and provide plausibly post-quantum secure mechanisms to prove the knowledge of an input to a public circuit that yields a specific output, due to a cut-and-choose approach over several runs. Our construction follows this approach to convince the service provider about the integrity of the user's data and its encryption. In a concurrent work~\cite{baum2020concretely}, Baum and Nof also present the use of MPC-in-the-head to prove lattice-based assumptions. However, their construction is based on a different problem (SIS, Short Integer Solution) and, unlike \syscomma, does not address the integrity check of the encrypted payload.

Another potential solution for enabling privacy-friendly and integrity preserving computations on authenticated data is to consider trusted hardware, such as Intel SGX~\cite{anati2013innovative,hoekstra2013using,mckeen2013innovative}, to process the data. The secure enclave could be positioned at the user side (returning a result certified by the enclave) or at the service provider side (decrypting ciphertexts and returning only the result of the computation). However, those solutions impose different trust assumptions and we consider them orthogonal to our work.

Several works are devoted to protecting privacy for smart metering (e.g., see surveys~\cite{erkin_privacy-preserving_2013,wang_cyber_2013}). However, only some of them, e.g.,~\cite{abdallah_lightweight_2018,li_preserving_2012}, address also the concern of data integrity and authenticity, by relying on custom homomorphic signature schemes. The applicability of such solutions is limited as, according to their technical specifications~\cite{SMETS2}, smart meters cope with standard digital signatures, e.g., ECDSA~\cite{nistHash}. Similarly, a number of works, e.g.,~\cite{ayday_protecting_2013,decristofaro2013secure,wang_healer:_2016}, employ homomorphic encryption to protect genomic privacy and to perform disease-susceptibility computations. Their model considers a medical unit that sequences the DNA of the user, who in turn protects it via homomorphic encryption before sending it for processing to a third-party. These solutions do not address the issue of data integrity or authenticity. Finally, several works are dedicated to both privacy and integrity in location-based activity tracking~\cite{zhu_applaus:_2011,wang_stamp:_2016,luo_veriplace:_2010,pham_securerun:_2016,saroiu2009enabling, hasan2005where}. They also are either peer-based~\cite{zhu_applaus:_2011,wang_stamp:_2016}, infrastructure-based~\cite{luo_veriplace:_2010,pham_securerun:_2016}, or hybrid\cite{saroiu2009enabling, hasan2005where}. SecureRun~\cite{pham_securerun:_2016} offers activity proofs for estimating the distance covered in a privacy and integrity preserving manner. Nevertheless, the system's accuracy relies on the density of access points, and it achieves at best a median accuracy of $78$\% (compared to $99.9$\% with \sysname on a similar dataset).

\section{Conclusion}\label{conc}

Data sharing among users and service providers in the digital era incurs a trade-off between privacy, integrity, and utility. In this paper, we have proposed a generic solution that protects the interests of both users and service providers. Building on state-of-the-art lattice-based homomorphic encryption and commitments, as well as zero-knowledge proofs, our construction enables users to offload their data to service providers in a post-quantum secure, privacy and integrity preserving manner, yet still enables flexible computations on it. We evaluated our solution on three different uses cases, showing its wide potential for adoption. As future work, we will explore extending \sysname to malicious service providers by combining secure computation techniques with differential privacy.


\section*{Acknowledgements}\label{ack}
We would like to thank our shepherd Ian Goldberg and the anonymous reviewers for their helpful feedback. We are also grateful to Henry Corrigan-Gibbs, Wouter Lueks, and the members of the EPFL Laboratory for Data Security for their helpful comments and suggestions. This work was supported in part by the grant \#2017-201 (DPPH) of the Swiss strategic focus area Personalized Health and Related Technologies (PHRT), and the grant C17-16 (SecureKG) of the Swiss Data Science Center.

\bibliographystyle{IEEEtranS}
\bibliography{literature}


\appendix

\section{Propositions and Lemmas}\label{app:lem}
For the sake of completeness, and for the reader's convenience, we present here the lemmas and propositions used in Section~\ref{analysis} to support the privacy and security analysis of our solution. 

\begin{lem}[Lemma\;1 in CKKS~\cite{cheon_homomorphic_2017}]\label{lem:ckks}
The encryption noise is bounded by $B_\text{clean}{=}8\sqrt{2}\sigma N{+}6\sigma \sqrt{N}{+}16\sigma\sqrt{hN}$. If $\bm{c}{\leftarrow} \text{Enc}_{\bm{pk}}(m)$ and $m{\leftarrow}\text{Ecd}(\bm{z},\Delta)$ for some $\bm{z}{\in}\mathbb{Z}[i]^{N/2}$ and $\Delta{>}N{+}2B_{clean}$, then $\text{Dcd}(\text{Dec}_{\bm{sk}}(\bm{c})){=}\bm{z}$.
\end{lem}
\noindent Ecd$(.,.)$ (resp. Dcd$(.)$) is the encoding (resp. decoding) function, $\Delta$ the scaling factor, and $h$ the Hamming weight of $\bm{sk}$.

\begin{lem}[Lemma\;4 in BDOP~\cite{baum_ecient_nodate}]\label{lem:bdop4}
Assume the distribution $\mathcal{D}$ and that $R_q$, m are chosen such that: 1) the min-entropy of a vector drawn from $\mathcal{D}$ is at least $(k + 1) \log |R_q| {+} \kappa$, where $\kappa$ is a (statistical) security parameter, and 2) the class of functions $\{f_a | a {\in} R_q^m\}$ where $f_a(r) {=} a \cdot r$ is universal when mapping the support of $\mathcal{D}$ to $R_q$. Then, the scheme is statistically hiding.
\end{lem}
\noindent Recall that $\mathcal{D}$ denotes the distribution of an honest prover’s randomness for commitments.

\begin{lem}[Lemma\;10 in BDOP~\cite{baum_ecient_nodate}]\label{lem:bdop10}
The protocol $\Pi_{\text{Bound}}$ has the following properties:

\begin{itemize}
\itemsep-0.5em 
    \item \textbf{Correctness:} The verifier always accepts an honest prover when the protocol does not abort. The probability of abort is at most $2/\gamma + 2/\gamma_x$.
    \item \textbf{Special soundness:} On input a commitment $\bm{c}$ and a pair of transcripts $(\bm{c},d,(c_{aux},\bm{t},z,\bm{r}_z))$, $(\bm{c},d',(c'_{aux},\bm{t'},z',\bm{r}'_{z}))$ where $d{\neq}d'$, we can extract either a witness for breaking the auxiliary commitment scheme, or a valid opening of $\bm{c}$ where the message $x$ has norm at most $\gamma_{x}N\beta_x$.
    \item \textbf{Honest-verifier zero-knowledge:} Executions of protocol $\Pi_{\text{Bound}}$ with an honest verifier can be simulated with statistically indistinguishable distribution.
\end{itemize}
\end{lem}
\noindent We remind that $\gamma$ is a constant that regulates the abort probability, $\gamma_x$ a constant piloting the norm of $\bm{z}$, and $\beta_x$ an upper bound on the norm of possible $x$.

\begin{lem}[Lemma\;3 in BDOP~\cite{baum_ecient_nodate}]\label{lem:bdop3}
From a commitment $\bm{c}$ and correct openings $\bm{r}$, $f$, $\bm{r'}$, $f'$ to two different message vectors $\bm{x}$, $\bm{x'}$, one can efficiently compute a solution $\bm{s}$ with $\|\bm{s}\|_{\infty} \leq 2 N m \gamma \beta \gamma_{D}^{2}$ to the Ring-SIS problem instance defined by the top row of $\bm{A}_1$.
\end{lem}

\begin{prop}[Proposition\;4.2 in ZKBoo~\cite{giacomelli_zkboo_2016}]\label{lem:zkboo}
The ZKBoo protocol is a $\Sigma$-protocol for the relation $R_\Phi$ with 3-special soundness.
\end{prop}

\section{ZKB++}\label{app:zkb}

\descr{(2,3)-decomposition of a Circuit from \cite{giacomelli_zkboo_2016, chase_post-quantum_2017}} For a function $f$ represented by a circuit $\mathcal{C}$, a (2,3)-decomposition consists of a series of algorithms $(\textsc{Share}, \textsc{Output}, \textsc{Rec}) \cup \textsc{Update}$, with $\textsc{Share}$ surjective, and allows to create three separate views of the circuit. Then, ZKB++ enables to prove knowledge of a secret input $x$ such that $f(x)=y$, with $y$ the publicly known output. For an iteration $k$, the views for player $i$ is denoted by a vector $\textbf{View}_i^k{=}\{\text{view}_i^{0},{\cdots},\text{view}_i^{N_g}\}$.
\begin{figure}[h!]
    \centering
    \includegraphics[width=.45\textwidth]{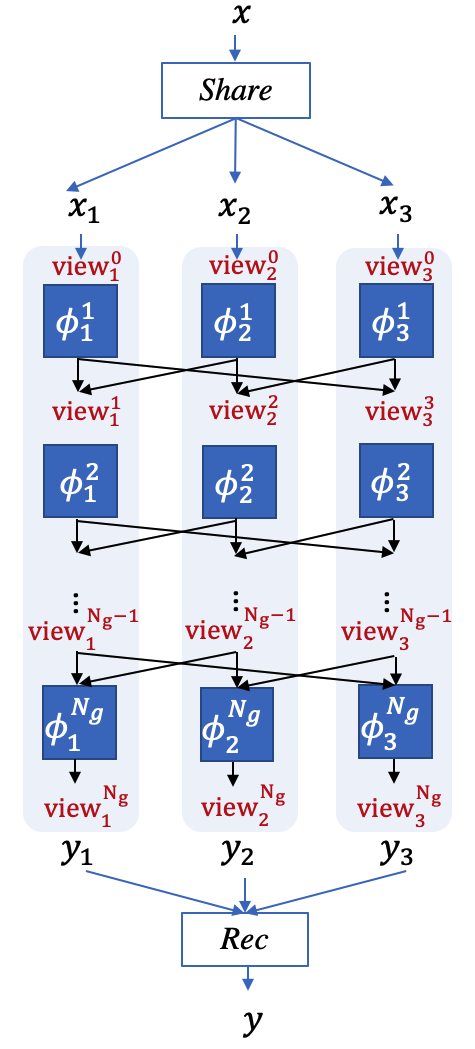}
    
    \caption{(2,3)-decomposition of a circuit.}
    \label{fig:zkboo}
    
\end{figure}

\begingroup
\setlength\abovedisplayskip{0.2em}
\setlength\belowdisplayskip{0.2cm}
\begin{itemize}\itemsep-0.5em 
\item The \textsc{Share} algorithm splits a secret $x$:
\[\left(\operatorname{view}_{1}^{0}, \operatorname{view}_{2}^{0}, \operatorname{view}_{3}^{0}\right) {=} \textsc{Share}(x, \bm{k}_1, \bm{k}_2, \bm{k_3}),\] 
with $\bm{k}_i$ being a random tape, $\forall i \in \{1, 2, 3\}$.
\item $\mathcal{F} {=} \bigcup\limits_{i=1}^{N_g} \{\phi_1^j, \phi_2^j, \phi_2^j\}$, where $N_g$ is the number of gates in $\mathcal{C}$ and $\phi_i^j$ is the $j$-th gate of player $i$.
\item The \textsc{Update} algorithm evaluates the gates into the views: 
$$\operatorname{view}_i^{j+1} {=}  \phi_i^j(\operatorname{view}_i^{j}, \operatorname{view}_{i+1}^{j}, \bm{k}_i, \bm{k}_{i+1}),$$
where $ j \in [0, N_g-1], \forall i \in \{1, 2, 3\}$.
\item The \textsc{Output} algorithm returns the output wires: \[y_i {=} \textsc{Output}(\operatorname{view}_i^{N_g}), \forall i \in \{1, 2, 3\}.\]
\item \textsc{Rec} reconstructs the output: $y {=} \textsc{Rec}(y_1, y_2, y_3).$
\end{itemize}

The intermediary functions $\phi_i^j$ are defined by running a linear decomposition on $\mathcal{C}$ such that: 
\begin{itemize}\itemsep-0.5em 
\item Each player $i$ has wire $\bm{w}^i$ and $w_{k}^{(1)}+w_{k}^{(2)}+w_{k}^{(3)}$ is equal to the wire state of the $k$-th gate of $\mathcal{C}$.
\item \textbf{Addition by a constant $d$:} $\forall i {\in} \{1,2,3\}$ \begin{equation*}
    w_{b}^{(i)}=\begin{cases}
     w_{a}^{(i)}+d & {\text{~if }  i=1,}\\ 
    {w_{a}^{(i)}} & {~\text{otherwise.}}
    \end{cases}
\end{equation*}
\item \textbf{Multiplication by a constant $d$:} $\forall i {\in} \{1,2,3\}$ $${w_{b}^{(i)}=d \cdot w_{a}^{(i)}}.$$
\item \textbf{Binary addition:} $\forall i \in \{1,2,3\}$ $${w_{c}^{(i)}=w_{a}^{(i)}+w_{b}^{(i)}}.$$
\item \textbf{Binary multiplication:} $\forall i \in \{1,2,3\}$ 
\[w_{c}^{(i)}=w_{a}^{(i)} \cdot w_{b}^{(i)}+w_{a}^{(i+1)} \cdot w_{b}^{(i)}+w_{a}^{(i)} \cdot w_{b}^{(i+1)}+R_{i}(c)-R_{i+1}(c),\]
where $R_i(c)$ is the $c$-th output of a pseudo-random generator seeded with $\bm{k}_i$.
\end{itemize}
The 2-privacy property ensures that revealing two views (i.e., opening two players) does not leak information about the witness. 
Overall, the ZKB++ protocol works as follows:
\begin{enumerate}[label=(\roman*)]\itemsep-0.5em 
\item The prover emulates three players. For each iteration $i \in [1, t]$, each player $j\in\{0,1,2\}$ evaluates the (2,3)-decomposition of the circuit and:

\begin{itemize}\itemsep-0.5em 
\item[--] commits to:
\[\left[C_{j}^{(i)}, D_{j}^{(i)}\right] {\leftarrow}\left[H \left(k_{j}^{(i)}, x_{j}^{(i)}, \textbf {View}_{j}^{i}\right), k_{j}^{(i)} \| \textbf {View}_{j}^{i}\right],\] 
\item[--] and lets $a^{(i)}=\left(y_{1}^{(i)}, y_{2}^{(i)}, y_{3}^{(i)}, C_{1}^{(i)}, C_{2}^{(i)}, C_{3}^{(i)}\right).$
\end{itemize}

\item The prover computes the challenge $e {=} H\left(a^{(1)}, {\ldots}, a^{(t)}\right)$ and reads it as a value $e^{(i)}{\in}\{1,2,3\}$, for all $i{\in} [1, t]$.
For all $i{\in} [1, t]$, the prover lets $b^{(i)}{=}\\\left(y_{e^{(i)}+2}^{(i)}, C_{e^{(i)}+2}^{(i)}\right)$ and \[z^{(i)} \leftarrow \left\{\begin{array}{ll}{\left(\textbf{View}_{2}^{(i)}, \bm{k}_{1}^{i}, \bm{k}_{2}^{i}\right)}& {\text { if } e^{(i)}=1,} \\ {\left(\textbf{View}_{3}^{(i)}, \bm{k}_{2}^{(i)}, \bm{k}_{3}^{(i)}, x_{3}^{(i)}\right)}& {\text { if } e^{(i)}=2,} \\ {\left(\textbf {View}_{1}^{(i)}, \bm{k}_{3}^{(i)},\bm{k}_{1}^{(i)}, x_{3}^{(i)}\right)}& {\text { if } e^{(i)}=3.}\end{array}\right. \]

\item Then, the prover computes the proof $$p {=} \left[e,\left(b^{(1)}, z^{(1)}\right),\left(b^{(2)}, z^{(2)}\right), {\cdots},\left(b^{(t)}, z^{(t)}\right)\right].$$
\item The verifier, for each iteration $i \in [1, t]$, reconstructs the input and output views that were not given as part of the proof by:
\begin{itemize}\itemsep-0.5em 
\item[--] running the circuit for player $e^i$ with the information in the proof, and
\item[--] computing $C_j^i, D_j^i$, and $a^i$, with the information provided in $b^i$.
\end{itemize}
\item Finally, the verifier computes the challenge ${e'{=}H\left(a^{(1)}, {\ldots}, a^{(t)}\right)}$ and checks that $e{\stackrel{?}{=}}e'$ is true.
\end{enumerate}
\endgroup

\section{Other Uses Cases}\label{app:uc}

As discussed in Section~\ref{intro}, the considered three-party model is generic enough to accommodate various use cases where simultaneously protecting privacy and integrity is of paramount importance. For instance, in business auditing, a company (i.e., the user) may have its financial books certified by a legal entity responsible for it, e.g., the director of the financial department (i.e., the data source) and a third-party auditing entity (i.e., the service provider) verifies their correctness in a privacy-preserving manner~\cite{backes_adsnark:_2015}. Loyalty programs constitute another use case that can be covered by the envisioned model. In this setting, customers of retailers obtain certifications about their purchases from payment terminals installed in shops by a bank and, to obtain discounts, use the certified data with the retailer ~\cite{blanco_privacy_2016}. Similarly, in personalized health applications, wearable devices act as data sources that sign users' blood pressure, heart rate, or quality of sleep. These data can be used with service providers that perform machine-learning operations to predict the user's health status and provide recommendations. Finally, other location-based activities, such as pay-as-you-drive insurance or dynamic road tolling~\cite{PrETP_2010}, are very similar to the location-based activity tracking case we evaluated in our work.

In Table~\ref{tab:uc}, we present the parties involved in the various use cases discussed in Section~\ref{intro}.

\begin{table}[H]
\footnotesize
\begin{center}
\begin{tabular}{ L{1in} L{1in} L{1in} }
\midrule
 \textbf{Use Case} & \textbf{Data Source} & \textbf{Service Provider} \\ 
\midrule
 Smart Metering & Smart Meter &  Energy Company billing and load balancing entities\\ 
\midrule
 Disease Susceptibility Test & Medical Unit or Sequencing Center & direct-to-consumer / pharmaceutical companies \\ 
\midrule
 Location-Based Activity Tracking & Cellular Network & Social Networks, Insurance Companies, etc.\\
\midrule
 Dynamic Road Tolling & Cellular Network & Tolling Company\\
\midrule
 Pay-as-you-drive Insurance & Cellular Network & Insurance Company\\
\midrule
 Business Auditing & Head of Financial Dep. & Auditing Company\\
\midrule
 Loyalty Program & Payment Terminal & Retailer \\
\midrule
\end{tabular}
\end{center}

\caption{Other examples of use cases.}
\label{tab:uc}
\end{table}

\section{Bound Proof of BDOP Commitments}\label{app:bdopBP}
Consider a given a commitment $\bm{c}=\text{BDOP}(m, \bm{r}_c)$ of a single input value $m$. The BDOP commitment scheme~\cite{baum_ecient_nodate} enables the prover to claim that the norm of $m$ is small through a bound proof. The protocol works as follows:
\begin{enumerate}\itemsep-0.5em 
    \item The prover computes a commitment $\bm{t} = \text{BDOP}(\mu, \bm{\rho})$ for $\mu$ sampled uniformly in $\mathcal{R}_q$ and $\bm{\rho} {\in} \mathcal{S}_\beta^{l_c}$ a valid commitment noise subjected to $\| \mu \|_{\infty}\leq \beta_\mu$ and $\| \bm{\rho} \|_{\infty}\leq \beta_\rho$. Then, the prover commits to this commitment through an auxiliary commitment and sends $c_{\text{aux}}{=}C_{\text{aux}}(\bm{t})$ to the verifier.
    \item The verifier randomly picks and sends a challenge $\text{d}{\in} \{0,1\}$.
    \item The prover checks the correctness of the challenge and computes $z{=}\mu+\text{d}m$ as well as $\bm{r}_z{=}\bm{\rho} + \text{d}\bm{r}_c$ in order to provide a valid opening of $\bm{t} {+} \text{d} \bm{c}$. The prover aborts if the resulting commitment is not properly formatted. Otherwise, it sends $z$, $\bm{r}_z$, and the opening of the auxiliary commitment to the verifier.
    \item The verifier checks the correctness of the opening and that the norm of $z$ is small. 
\end{enumerate}
At the end of the protocol, the verifier is ensured that the norm of the secret $m$ is below a specific threshold. This protocol can be made non interactive with the Fiat-Shamir heuristic and iterated to increase the soundness. We refer the reader to~\cite{baum_ecient_nodate} for more details.


\end{document}